%% file: main.tex
\documentclass[11pt,letter]{article}
\input{header.tex}

\begin{document}
 
\linespread{0.97}

\title{Tight Bounds for Approximate Carath\'{e}odory and Beyond}

\author[1]{Vahab Mirrokni}
\author[1]{Renato Paes Leme}
\author[2]{Adrian Vladu}
\author[3]{Sam Chiu-wai Wong}
\affil[1]{Google Research NY, \texttt{\{mirrokni, renatoppl\}@google.com}}
\affil[2]{MIT Math, \texttt{avladu@mit.edu}}
\affil[3]{UC Berkeley, \texttt{samcwong@berkeley.edu}}

\date{\today}
\date{}
\maketitle
\begin{abstract}

We give a deterministic nearly-linear time algorithm for approximating any point inside a convex polytope with a sparse convex combination of the polytope's vertices. Our result provides a constructive proof for the Approximate Carath\'{e}odory Problem~\cite{Bar15}, which states that any point inside a polytope contained in the $\ell_p$ ball of radius $D$ can be approximated to within $\epsilon$ in $\ell_p$ norm by a convex combination of only $O\left(D^2 p/\epsilon^2\right)$ vertices of the polytope for $p \geq 2$. We also show that this bound is tight, using an argument based on anti-concentration for the binomial distribution.

Along the way of establishing the upper bound, we develop a technique for minimizing norms over convex sets with complicated geometry; this is achieved by running Mirror Descent on a dual convex function obtained via Sion's Theorem.

As simple extensions of our method, we then provide new algorithms for submodular function minimization and SVM training. For submodular function minimization we obtain a simplification and (provable) speed-up over Wolfe's algorithm, the method commonly found to be the fastest in practice. For SVM training, we obtain $O(1/\epsilon^2)$ convergence for arbitrary kernels; each iteration only requires matrix-vector operations involving the kernel matrix, so we overcome the obstacle of having to explicitly store the kernel or compute its Cholesky factorization.

\end{abstract}

\thispagestyle{empty}

\newpage
\setcounter{page}{1}
\input{intro.tex}
\input{prelim.tex}
\input{lineartime.tex}

\input{lower.tex}

\input{app.tex} 
\bibliographystyle{plain}
\footnotesize
\bibliography{biblio}
\end{document}

%% file: header.tex
\usepackage[papersize={8.5in,11in},margin=1in]{geometry}
\usepackage{amsthm} 
\usepackage{fancyhdr} % extensive control of page headers and footers in LaTeX2e
\usepackage{authblk}

\usepackage{color} % color control for LaTeX documents
\definecolor{note_fontcolor}{rgb}{0.800781, 0.800781, 0.800781}
\definecolor{ForestGreen}{rgb}{0.1333,0.5451,0.1333}

\usepackage{changepage}
\usepackage[nofillcomment,linesnumbered,noend,noresetcount,noline]{algorithm2e}

\AtBeginDocument{%
 \abovedisplayskip=2pt plus 1pt minus 1pt
 \abovedisplayshortskip=0pt plus 1pt
 \belowdisplayskip=2pt plus 1pt minus 1pt
 \belowdisplayshortskip=0pt plus 1pt
}
\usepackage[noindentafter]{titlesec}
\titlespacing{\paragraph}{%
  0pt}{%              left margin
  0.1\baselineskip}{% space before (vertical)
  1em}%               space after (horizontal)
\titlespacing\section{0pt}{8pt plus 1pt minus 1pt}{2pt plus 1pt minus 1pt}
\titlespacing\subsection{0pt}{8pt plus 1pt minus 1pt}{2pt plus 1pt minus 1pt}
\titlespacing\subsubsection{0pt}{8pt plus 1pt minus 1pt}{2pt plus 1pt minus 1pt}

\usepackage{amsmath} % enhances the typeset appearance of mathematical formulas
\usepackage{amssymb} % subset of amsfonts that defines the full set of symbol names for two fonts of extra symbols included in amsfonts
\usepackage{graphicx} % enhanced support for graphics
\usepackage{complexity} % computational complexity class names
\usepackage{enumitem} % control layout of itemize, enumerate, description
\usepackage{shorttoc} % table of contents with different depths
\usepackage{array} % extending the array and tabular environments
\usepackage{float}
\usepackage{pdfsync}
\usepackage{verbatim}
\usepackage[pagebackref=true,colorlinks]{hyperref}
\hypersetup{linkcolor=red,filecolor=red,citecolor=ForestGreen,urlcolor=red}

\PassOptionsToPackage{normalem}{ulem}
\usepackage{ulem}

\newtheoremstyle{slplain}% name
  {.4\baselineskip\@plus.1\baselineskip\@minus.1\baselineskip}% Space above
  {.3\baselineskip\@plus.1\baselineskip\@minus.1\baselineskip}% Space below
  {\itshape}% Body font
  {}%Indent amount (empty = no indent, \parindent = para indent)
  {\bfseries}%  Thm head font
  {.}%       Punctuation after thm head
  { }%      Space after thm head: " " = normal interword space;
  {}%       Thm head spec
\theoremstyle{slplain} % italics

\newtheorem*{definition*}{Definition}
\newtheorem*{theorem*}{Theorem}
\newtheorem{theorem}{Theorem}[section]
\newtheorem{lemma}[theorem]{Lemma}
\newtheorem{proposition}[theorem]{Proposition}

\newtheorem{corollary}[theorem]{Corollary}
\newtheorem{definition}[theorem]{Definition}

\makeatletter
\newtheorem*{rep@theorem}{\rep@title}
\newcommand{\newreptheorem}[2]{%
\newenvironment{rep#1}[1]{%
 \def\rep@title{#2 \ref{##1}}%
 \begin{rep@theorem}}%
 {\end{rep@theorem}}}
\makeatother

\newreptheorem{theorem}{Theorem}
\newreptheorem{lemma}{Lemma}
\newreptheorem{claim}{Claim}
\newreptheorem{corollary}{Corollary}
\newreptheorem{proposition}{Proposition}

\theoremstyle{definition}

\newtheorem{remark}[theorem]{Remark}

\theoremstyle{plain} % choose from plain/definition/remark

\numberwithin{equation}{section}

\newtheoremstyle{etplain}% name
  {.0\baselineskip\@plus.1\baselineskip\@minus.1\baselineskip}% Space above
  {.0\baselineskip\@plus.1\baselineskip\@minus.1\baselineskip}% Space below
  {\itshape}% Body font
  {}%Indent amount (empty = no indent, \parindent = para indent)
  {\bfseries}%  Thm head font
  {.}%       Punctuation after thm head
  { }%      Space after thm head: " " = normal interword space;
  {}%       Thm head spec

\newcommand{\litlow}[1]{\mathord{\mathcode`\-="702D\sf #1\mathcode`\-="2200}}
\newcommand{\lit}[1]{\ensuremath{\litlow{#1}}}

\newcommand{\namedref}[2]{\hyperref[#2]{#1~\ref*{#2}}}

\newcommand{\figurerefb}[2]{\hyperref[#1]{Figure~\ref*{#1}#2}}

\newcommand{\equationref}[1]{\hyperref[#1]{(\ref*{#1})}}
\renewcommand{\eqref}{\equationref}
\usepackage[textsize=tiny]{todonotes}

\newcommand{\DEBUG}[1]{}

\newcommand{\argmax}{\operatornamewithlimits{arg\,max}}
\newcommand{\argmin}{\operatornamewithlimits{arg\,min}}

\newcommand{\hull}{\text{\textnormal{conv}}}
\newcommand{\pow}{\text{\textsc{Pow}}}
\newcommand{\diag}{\text{\textsc{Diag}}}
\newcommand{\one}{\vec{1}}

\def\abs#1{\left|#1  \right|}
\def\norm#1{\left\| #1 \right\|}

\def\ssupp#1{\left\vert\textnormal{supp}\left(#1\right)\right\vert}
\newcommand\BB{\boldsymbol{\mathit{B}}}

\newcommand\sgn[1]{\textnormal{sgn}(#1)}

\newcommand\RCH{\textnormal{\textsc{Rch}}} %textbf inside textnormal so that it's displayed the same even in thm environments
\renewcommand\R{\mathbb{R}}
\renewcommand\E{\mathbb{E}}
\renewcommand\P{\mathbb{P}}
\newcommand\supp{\textnormal{supp}}

%% file: intro.tex
%!TEX root = main.tex

\section{Introduction}
The (exact) Carath\'{e}odory Theorem is a fundamental result in convex geometry
which states that any point $u$ in a polytope $P\subseteq\mathbb{R}^{n}$
can be expressed as a convex combination of $n+1$ vertices of $P$.
Recently, Barman~\cite{Bar15} proposed an approximate version and showed that it can
be used to improve algorithms for computing Nash equilibria in game theory and
algorithms for the $k$-densest subgraph in combinatorial optimization. Versions
of the Approximate Carath\'{e}odory Theorem have been proposed and applied in
different settings. Perhaps its most famous incarnation is as  Maurey's
Lemma~\cite{pisier1980remarques} in functional analysis. It states
that if one is willing to tolerate an error of $\epsilon$ in $\ell_{p}$
norm, $O\left(D^{2}p/\epsilon^{2}\right)$ vertices suffice to approximate
$u$, where $D$ is the radius of the smallest $\ell_{p}$ ball enclosing $P$.
The key significance of the approximate Carath\'{e}odory Theorem is that
the bound it provides is \emph{dimension-free}, and consequently
allows us to approximate any point inside the polytope with a \emph{sparse}
convex combination of vertices.

\begin{center}
\fbox{\begin{minipage}[t]{1\columnwidth}%
\uline{The Approximate Carath\'{e}odory Problem}

Given a polytope $P$ contained inside the $\ell_{p}$ ball of radius
$D$, and $u\in P$, find a convex combination $\sum_{i=1}^{k}x_{i}v_{i}$
of vertices $v_{i}$ of $P$ such that $k = O\left(D^{2}p/\epsilon^{2}\right)$
and $\left\Vert \sum_{i=1}^{k}x_{i}v_{i}-u\right\Vert _{p}\leq\epsilon$.%
\end{minipage}}
\par\end{center}

Both Barman's proof and Maurey's original proof start from a solution
$u=\sum_{i=1}^{n+1}\lambda_{i}v_{i}$ of the exact Carath\'{e}odory problem,
interpret the coefficients $\lambda_i$ of the convex combination as a
probability distribution and generate a sparse solution by sampling from the
distribution induced by $\lambda$. Concentration inequalites are then used to
argue that the average sampled solution is close to $u$ in the $\ell_p$-norm. 
The proof is clean and elegant, but it leaves two questions: Is randomization
really necessary for this proof? And, can we bypass a solution to the exact
approximate Carath\'{e}odory problem and directly compute a solution to the
approximate version? 

The second question is motivated by the fact that 
computing the solution to the exact Carath\'{e}odory problem can be costly. In
fact, this takes $O(n^\omega)$ time even if the points $v_i$ 
are known in advance. The situation becomes even worse for 
polytopes for which it is not desirable to maintain an
explicit representation of all its vertices (e.g. the
matching polytope or the matroid base polytope) since
there may be exponentially many of them. In this
case, even finding the $n+1$ vertices whose convex hull contains $u$ becomes
significantly more difficult. 

Our first contribution addresses those two questions by giving 
a constructive proof of the approximate Carath\'{e}odory Theorem. As a
corollary, this gives the first nearly linear time deterministic algorithm
for the approximate Carath\'{e}odory problem that does not require knowing
$u=\sum_{i=1}^{n+1}\lambda_{i}v_{i}$ in advance. Our algorithm runs
in $O(D^{2}p/\epsilon^{2})$ iterations, each of which takes linear
time.

Our second contribution is to provide a lower bound showing that the
$O(D^{2}p/\epsilon^{2})$ factor is tight. This improves upon a lower bound of
$\Omega((D/\epsilon)^{p/(p-1)})$ proved by Barman. Barman's lower bound is tight
up to constant factors for $\ell_2$ but leaves a significant gap for any $p >
2$. We prove our lower bound by exhibiting a
polytope $P$ in the radius-$D$ $\ell_p$ ball and a point $u$ inside for which all
convex combinations of $O( D^2 p / \epsilon^2)$ vertices are $\epsilon$-far from the $u$
in the $\ell_p$-norm.

These are in principle the best results one can hope for. We also show that even
though the dependence on $\epsilon$ can't be improved in general, it can be
greatly improved in a special case. If $u$ is far away from the boundary of $P$,
i.e., if the ball of radius $r$ around $u$ is contained in $P$, then there exist
a solution to the approximate Carath\'{e}odory problem with $k = O\left(
\frac{D^2 p}{r^2} \log\left( \frac{r}{\epsilon} \right) \right)$.

In order to achieve the positive results for approximate Carath\'{e}odory, we develop a technique 
for minimizing norms over convex sets with complicated geometry; this is achieved by running {\em Mirror Descent} on 
a dual convex function obtained via Sion's Theorem.  This technique may be of independent interest.
To show its potential, we note that simple extensions of our method result in new algorithms for submodular function minimization and SVM training. For submodular function minimization, we obtain a simplification and (provable) speed-up over Wolfe's algorithm, the method commonly found to be the fastest in practice. For SVM training, we obtain $O(1/\epsilon^2)$ convergence for arbitrary kernels; each iteration only requires matrix-vector operations involving the kernel matrix, so we overcome the obstacle of having to explicitly store the kernel or compute its Cholesky factorization. Next, we elaborate on our technique and then discuss these applications in
more details.

\subsection{Our techniques: mirror descent and lower bounds}

Our new constructive proof of the approximate Carath\'{e}odory Theorem
employs a technique from Convex Optimization called \emph{Mirror Descent}, which
is a generalization of \emph{Subgradient Descent}. Both subgradient and mirror descent
are first order methods that minimize arbitrary convex functions to an additive
precision of $\epsilon$ using only information about the subgradient of the
function.

In particular, we formulate the approximate Carath\'{e}odory
problem as:
\[
\min_{x\in\Delta}\left\Vert Vx-u\right\Vert _{p}
\]
where the columns of $V$ are the vertices of $P$ and $\Delta$
is the unit simplex.

Our first instinct is to apply gradient or mirror descent to $f(x)=\left\Vert
Vx-u\right\Vert _{p}$. This fails to achieve any sort of sparseness guarantee
since  ts gradient is generally not sparse and the new iterate 
$x_{t+1}=x_{t}-\eta\nabla f(x_{t})$
would not be either.

Inspired by algorithms for solving positive linear programs such as
\cite{PST91,young01}, we reformulate our problem as a \emph{saddle point
problem} $\max_{y \in \BB_q(1)} \min_{x \in \Delta} y^\top (Vx-u)$, where
$\BB_q(1)$ is the $\ell_q$-ball. This can be viewed as a zero-sum game. Applying
a generalization of the minimax theorem we can obtain a \emph{dual convex
function}.

We apply the mirror descent framework to this dual function. Mirror descent is a
framework that needs to be instantiated by the choice of a \emph{mirror map},
which plays a similar role as linking functions in Online Learning. We will
provide an overview of Mirror Descent in the next section so that the paper is
self contained. But for the reader familiar with Mirror Descent terminology, our
mirror map is a truncated version of the square $\ell_q$-norm, where $q$ is chosen
such that $\frac{1}{p} + \frac{1}{q} = 1$.

The analysis is enabled by choosing the right function to optimize and the
appropriate mirror map.
At the very high level, our algorithm incrementally improves our current
choice of $x$ by expanding its support by $1$ in each iteration. The
desired sparsity of $x$ then follows as we can show that the number
of iterations is $O(D^{2}p/\epsilon^{2})$.

Our lower bound is inspired by a method proposed by Klein and Young
\cite{KleinY15} for proving conditional lower bounds on the running time for solving
positive linear programs. It again follows from interpreting the
Carath\'{e}odory problem as a zero-sum game between a maximization and a
minimization player. We construct a random instance such that with high
probability the minimization player has a dense strategy with value close to zero,
but for every sparse support, the maximization player can force the strategy to
be $\epsilon$-far from zero with high probability. The lower bound follows from
taking the union bound over the probabilities and applying the probabilistic
method.

\subsection{Applications}
In the following, we discuss a number of applications of our results and techniques. 
While the first result is a straightforward use of our improved 
approximate Carath\'{e}odory theorem, the second result is a simple application
of the mirror-descent technique, and the third one is a simple application of 
an extension of the technique to SVMs. 

\paragraph{Warm-up: fast rounding in polytopes with linear optimization oracles.}

The most direct application of our approach is to efficiently round
a point in a polytope whenever it admits a good linear optimization
oracle. An obvious such instance is given by the matroid polytope.
Given an $n$-element matroid by $\mathcal{M}$ of rank $r$ and a
fractional point $x^{*}$ inside its base polytope, our algorithm produces
a sparse distribution $\mathcal{D}$ over matroid bases such that
marginals are approximately preserve in expectation. More specifically,
for any $p\geq2$, $\mathcal{D}$ has a support of size $\frac{p\cdot r{}^{2/p}}{\epsilon^{2}}$,
and $\left\Vert \mathbb{E}_{x\sim\mathcal{D}}\left[x\right]-x^{*}\right\Vert _{p}\leq\epsilon$;
furthermore, computing $\mathcal{D}$ requires only $O\left(nr^{2/p}p/\epsilon^{2}\right)$
calls to $\mathcal{M}$'s independence oracle. 

\paragraph{Submodular function minimization.}

Fujishige's \emph{minimum-norm point} algorithm is the method typically employed
by practitioners, to minimize submodular
functions~\cite{fujishige2011submodular,bach2010convex}.
Traditionally this has been
implemented using variants of Wolfe's algorithm~\cite{Wolfe76}, which lacked a rigorous
convergence analysis (it was only known to converge in exponential number
of steps). Only recently Chakrabarty, Jain, and Kothari~\cite{ChakrabartyJK14} proved the first polynomial
time bound for this method, obtaining an algorithm that runs in time
$O\left(\left(n^{5}\cdot\mathcal{T}+n^{7}\right)F^{2}\right)$, where
$\mathcal{T}$ is the time required to answer a single query to $f$,
and $F$ is the maximum marginal difference in absolute value.

As our second application, we show that our technique can replace Wolfe's
algorithm in the analysis of \cite{ChakrabartyJK14} obtaining an 
an $O\left(n^{6}F^{2}\mathcal{T}\right)$ time algorithm for
exact sumbodular function minimization, and a $O(n^{6}F^{2}\mathcal{T}/k^{4})$
for a $k$-additive approximation. We emphasize
that those are not the best theoretical algorithm, but
a simplification and a speed-up of the algorithm that is commonly found to be
the fastest in practice.

\paragraph{Support vector machines.}
Training support vector machines (SVMs) can also be formulated as minimizing a
convex function. We show that our technique of converting a problem to a saddle
point formulation and solving the dual via Mirror Descent can be applied to the
problem of training $\nu$-SVMs. This is based on a formulation introduced by 
 Sch{\"{o}}lkopf, Smola, Williamson, and Bartlett~\cite{ScholkopfSWB00}.
 Kitamura, Takeda and Iwata~\cite{KitamuraTI14} show how SVMs can be trained using Wolfe's
 algorithm. Replacing Wolfe's algorithm by Mirror Descent we obtain 
an $\epsilon$-approximate solution in time $O\left(\max\left(\frac{1}{\epsilon},\left\Vert K\right\Vert \right)/\left(\nu n\epsilon^{2}\right)\right)$,
where $K$ is the kernel matrix. Whenever the empirical data belongs
to the unit $\ell_{2}$ ball, this yields a constant number of iterations
for polynomial and RBF kernels. Our method does not need to explicitly
store the kernel matrix, since every iteration only requires a matrix-vector
multiplication, and the entries of the matrix can be computed on-the-fly
as they are needed. In the special case of a linear kernel, each iteration
can be implemented in time linear in input size, yielding a nearly-linear
time algorithm for linear SVM training.

\subsection{Related work}

As previously mentioned, the Approximate Carath\'{e}odory Theorem was been
independently discovered
many times in the past. The earliest record is perhaps due to Novikoff
\cite{novikoff1962convergence} in 1962 who showed that the $\ell_2$ version of
Approximate Carath\'{e}odory can be obtained as a byproduct of the analyis of
the Perceptron Algorithm (as pointed out by \cite{blum2015sparse}). Maurey
\cite{pisier1980remarques} proves it in the context of functional analysis.
We refer to
the appendix of \cite{BourgainN13} for the precise statement of Maurey's lemma
as well as a self-contained proof.
Farias et al ~\cite{farias2012sparse} studies it for the special case of the the
bipartite matching polytope. Barman~\cite{Bar15} study the $\ell_p$ case and
provides several applications to game theory and combinatorial optimization.

Related to the Approximate Carath\'{e}odory problem is the question studied by
Shalev{-}Shwartz, Srebro and Zhang of minimizing the loss of a linear predictor
while bounding the number of features used by the predictor. Their main result
implies a gradient-descent based algorithm for the $\ell_2$-version of the
Approximate Carath\'{e}odory Theorem but is only able to produce
$O(1/\epsilon^p)$ for $p > 2$. A different optimization approach to  Approximate
Carath\'{e}odory is done by Garber and Hazan \cite{GarberH13} who solve the
optimization problem $\min_{x \in P} \norm{x-u}_2^2$ using Frank-Wolfe methods,
also obtaining the $\ell_2$ version of the result.

Finally, the literature on Mirror Descent is too large to survey, but we refer
to the book by Ben-Tal and Nemirovski \cite{Nemirovski} for a comprehensive
overview, including a discussion of the $\ell_q$ square mirror map. In Online
Learning a variant of this mirror map has been used in Gentile's $p$-norm
algorithms \cite{Gentile03a}.

%% file: prelim.tex
%!TEX root = main.tex
\section{Preliminaries}

\subsection{Notation}

Given a point $x \in \R^d$, we define its $\ell_p$-norm as $\norm{x}_p =
\left(\sum_{i=1}^d \abs{x_i}^p\right)^{1/p}$ for $1\leq p < \infty$ and its
$\ell_\infty$ norm by $\norm{x}_\infty = \max_i \abs{x_i}$. Given a norm 
$\norm{\cdot}$, we denote by
$\BB_{\norm{\cdot}}(r) = \{x \in \R^d; \norm{x} \leq r\}$. For $\ell_p$ and
$\ell_\infty$ norms, we denote the balls simply by $\BB_p(r)$ and
$\BB_\infty(r)$.

Given a norm $\norm{\cdot}$, we define its dual norm $\norm{\cdot}_*$ as
$\norm{y}_* = \max_{x;\norm{x}=1} y^\top x$ in such a way that  H\"{o}lder's
inequality holds with equality: $y^\top x \leq \norm{y}_* \cdot \norm{x}$. The dual norm of
the $\ell_p$ norm is the $\ell_q$ norm for $\frac{1}{p} + \frac{1}{q} = 1$.

Given a vector $x$, let its support $\supp(x) = \left\{i | x_i \neq 0\right\}$ represent the number
of nonzero coordinates of $x$.

\subsection{Approximate Carath\'{e}odory problem}

The (exact) Carath\'{e}odory Theorem is a fundamental result in linear algebra which
bounds the maximum number of points necessary to describe a point in the convex
hull of a set. More precisely, given a finite set of points $X \subseteq \R^d$
and a point 
$u \in \hull(X) := \{\sum_{x \in X} \lambda_x \cdot x; \sum_x \lambda_x = 1, \lambda_x
\geq 0\}$, there exist $d+1$ points in $x_1, \hdots, x_{d+1} \in X$ such that $u
\in \hull\{x_1, \hdots, x_{d+1}\}$. On the plane, in particular, every point in
the interior of a convex polygon can be written as a convex combination of three
of its vertices.

The approximate version of the Carath\'{e}odory theorem bounds the number of points
necessary to describe a point $u \in \hull(X)$ approximately. Formally, given a norm
$\norm{\cdot}$, an additive error parameter $\epsilon$ and a set of points
$X \subseteq B_{\norm{\cdot}}(1) \subseteq \R^d$, for every $u \in
\hull(X)$ we want $k$ points $x_1, \hdots, x_k \in X$ such that there exists 
$u' \in \hull \{x_1,\hdots, x_k\}$ and $\norm{u-u'} \leq \epsilon$.

A general result of this type is given by Maurey's Lemma \cite{pisier1980remarques}. For the
case of $\ell_p$ norms, $p \geq 2$, Barman \cite{Bar15} showed that
$k \leq 4p/\epsilon^2$ points suffice. A notable aspect of this theorem
is that the bound is independent of the dimension of the ambient space.

\subsection{Convex functions}

We give a brief overview on the theory of convex functions. For a detailed
exposition we refer readers to \cite{rockafellar}.

\paragraph{Subgradients.} A function 
$f : Q \subseteq \R^d \rightarrow \R$ defined on
a convex domain $Q$ is said to be convex if every point $x \in Q$ has a
non-empty subgradient $\partial f(x) = \{ g \in \R^d; f(y) \geq f(x) +
g^\top(y-x), \forall y \in Q\}$. Geometrically, this means that a function is
convex iff it is the maximum of all its supporting hyperplanes, i.e. $f(x) =
\max_{x_0, g \in \partial f(x_0)} f(x_0) + g^\top (x-x_0)$.
When there is a unique element in $\partial
f(x)$ we call it the gradient and denote it by $\nabla f(x)$. We will sometimes abuse
notation and refer to $\nabla f(x)$ as an arbitrary element of $\partial f(x)$
even when it is not unique.

\paragraph{Strong convexity and smoothness.} We say that a function 
$f:Q \subseteq \R^d \rightarrow \R$ is $\mu$-strongly convex with respect to
norm $\norm{\cdot}$ if for all $x,y \in Q$ and all subgradients
$g \in \partial f(x)$:
$$f(y) - f(x) - g^\top (y-x) \geq \frac{1}{2} \mu \norm{y-x}^2$$
A function is said to be $\sigma$-smooth with respect to the $\norm{\cdot}$ 
if for all $x,y \in Q$ and $g \in \partial f(x)$:
$$f(y) - f(x) - g^\top (y-x) \leq \frac{1}{2} \sigma \norm{y-x}^2$$

\paragraph{Bregman divergence and the Hessian.} Every continuously
differentiable $f$ induces a  concept of `distance' known as the Bregman-divergence: 
given $x,y \in Q$, we define $D_f(y \Vert x) := f(y) - f(x)
- \nabla f(x)^\top (y-x)$ as the second order error when computing $f(y)$ using
the linear approximation of $f$ around $x$. The fact that $f$ is convex
guarantees $D_f(y \Vert x) \geq 0$.

If the subgradient of $f$ is unique everywhere, we can define $\mu$-strong
convexity and $\sigma$-smoothness with respect to the Bregman divergence, as
$D_f(y \Vert x) \geq \frac{1}{2} \sigma \norm{x-y}^2$ and $D_f(y \Vert x) \leq
\frac{1}{2} \mu \norm{x-y}^2$. If $f$ is also twice-differentiable,
a simple way to compute its strong convexity and 
smoothness parameters is by bounding the $\norm{\cdot}$-eigenvalues of the 
Hessian. If $\mu \cdot \norm{w}^2 \leq w^\top \nabla^2 f(x) w \leq \sigma \cdot
\norm{w}^2$ for all $x \in Q$ and $w \neq 0$, then $f$ is $\mu$-strongly convex 
and $\sigma$-smooth. This is because:
$$D_f(y \Vert x) = \int_0^1 [\nabla f(x + (y-x)t) - \nabla f(x)]^\top (y-x) dt = \int_0^1
\int_0^s (y-x)^\top \nabla^2 f(x + (y-x)s) (y-x) ds dt$$

\paragraph{Lipschitz constant.}
We say that a convex function is $\rho$-Lipschitz with respect to
norm $\norm{\cdot}$ if $\norm{\nabla f(x)}_* \leq \rho$. 
Note that $\rho$-Lipschitz continuity requires a bound on the
\emph{dual} norm, since $$\begin{aligned}
\abs{f(y) - f(x)} & = \abs{\int_0^1 
\nabla f(x+t(y-x))^\top (y-x) dt } \leq \int_0^1 \abs{\nabla f(x+t(y-x))^\top (y-x)}
dt \\ & \leq \int_0^1 \norm{\nabla f(x+t(y-x))}_* \cdot \norm{y-x} dt \leq \rho \cdot
\norm{y-x} \end{aligned}$$

\paragraph{Fenchel duality.}
It is useful to write a convex function as the maximum of its supporting
hyperplanes. One way to do that is using the \emph{Fenchel transform}.
When defining Fenchel transforms, it is convenient to identify a function $f:Q
\rightarrow \R$ to its extension $\tilde{f}:\R^d \rightarrow \R \cup \{\infty\}$
such that $\tilde{f}(x) = f(x)$ for $x \in Q$ and $\tilde{f}(x) = \infty$
otherwise. Given that identification, we can define the Fenchel transform of a
function $f$ as the function $f^*:\R^d \rightarrow \R \cup \{\infty\}$  given by
$f^*(z) = \sup_{x \in \R^d} z^\top x - f(x)$. If $f$ is convex, the Fenchel
transformation is self-invertible, i.e., $(f^*)^* = f$ or 
equivalently: $f(x) = \max_z z^\top x - f^*(z)$. Notice that the previous
expression is a way to write any convex function as a maximum over linear
functions in $x$ parametrized by $z$. The \emph{Fenchel inequality} $f(x) +
f^*(z) \geq z^\top x$ follows directly from the definition of the Fenchel
transform.

\paragraph{Envelope Theorem.} When writing a convex function $f = \max_i f_i$
as a maximum of other convex function (typically linear functions),
the Envelope Theorem gives a way to compute derivatives. Its statement is
quite intuitive: since gradients are local objects, the gradient of $f$ at a
certain point is the gradient of the function $f_i$ being maximized at that
point. Formally, if $f(x) = \max_z g(x,z)$ where $g(x,z)$ is convex in $x$ for
every fixed $z$, then if $f(x_0) = g(x_0, z_0)$, then $\partial_x g(x_0, z_0)
\subseteq \partial f(x_0)$. A direct application of this theorem is in computing
the gradients of the Fenchel dual: $\nabla f^*(z) = \argmin_x \{z^\top x -
f(x)\}$ and $f^*(z) = z^\top \nabla f^*(z) - f(\nabla f^*(z))$.

\paragraph{Smoothness and strong convexity duality.} Finally, we will use the
following duality theorem:
\begin{theorem}\label{thm:smoothness_strong_convexity}
The function $f : Q \rightarrow \R$ is a $(1/\sigma)$-strongly convex
function with respect to $\norm{\cdot}$ if and only if its Fenchel dual
$f^* : \R^d \rightarrow \R$ is a $\sigma$-smooth with respect to
$\norm{\cdot}_*$.
\end{theorem}

\begin{proof}
Here  we prove that $\sigma$-strong convexity of a function implies
$(1/\sigma)$-smoothness of its dual, since this is the direction we will use.
We refer to \cite{Kakade12,ShalevThesis} for a proof of the converse.

Fix $z_1,z_2 \in \R^d$  and let $y_i \in \partial f^*(z_i) = \argmax_{y \in Q}
z_i^\top y - f(y)$. Since $f$ is strongly convex, there in an unique maximum, so
we can write $y_i = \nabla f^*(z_i)$.
Also, $f^*(z_i) = z_i^\top y_i - f(y_i)$. Since
the Fenchel transform is self-dual, $f(y_i) = \max_z y_i^\top z - f^*(z) =
z_i^\top y_i - f^*(z_i)$. In particular, this means that $z_i \in \partial
f(y_i)$.

Using the strong-convexity of $f$, we can write:
$$f(y_2) - f(y_1) - z_1^\top(y_2-y_1) \geq \frac{1}{2\sigma} \norm{y_1 - y_2}^2 $$
$$f(y_1) - f(y_2) - z_2^\top(y_1-y_2) \geq \frac{1}{2\sigma} \norm{y_1 - y_2}^2 $$
Summing the expressions above and applying Holder's inequality, we get:
$$\frac{1}{\sigma} \norm{y_1-y_2}^2 \leq (z_2 - z_1)^\top(y_2 - y_1) \leq
\norm{z_1 - z_2}^* \cdot \norm{y_1 - y_2} $$
Therefore:
$$\sigma \cdot \norm{z_1 - z_2 }_* \geq \norm{y_1 - y_2} =
\norm{ \nabla f^*(z_1) - \nabla f^*(z_2) }$$
which implies the smoothness bound:
$$D_{f^*}(z_2 \Vert z_1) = \int_0^1 [\nabla f^*(z_1 + t(z_2 - z_1)) - \nabla
f^*(z_1)]^\top (z_2-z_1) dt \leq \frac{1}{2} \sigma \norm{z_1 - z_2}_*^2$$
\end{proof}

\subsection{A primer on Mirror Descent}

For the sake of completeness, we will present here an elementary exposition
of the Mirror Descent Framework, which is used in our proof.
For a complete exposition we refer to Nemirovskii \cite{Nemirovski}
or Bubeck  \cite{Bubeck14}.

The goal of Mirror Descent is to minimize a convex function
$f:Q \subseteq \R^d \rightarrow \R$ with Lipschitz constant $\rho$ with respect to norm
$\norm{\cdot}$. To motivate Mirror Descent, it is useful to think of dot
products $y^\top x$ as a product of vectors in two different vector spaces,
which can be thought as \emph{vectors} vs \emph{linear forms} or \emph{column
vectors} vs \emph{row vectors}. In the spirit of H\"{o}lder's
inequality, we can think of $x$ as living in the $\R^d$ space equipped
with $\norm{\cdot}$ norm while $y$ lives in $\R^d$ equipped with the
dual norm $\norm{\cdot}_*$. When we approximate $f(y) - f(x) \approx \nabla
f(x)^\top (y-x)$, the second term is a dot-product of a vector in the domain 
$y-x$, which we call the \emph{primal space} and measure using $\norm{\cdot}$ norm
and a gradient vector, which we call the \emph{dual space} and measure with
dual norm $\norm{\cdot}_*$.

Keeping the discussion in the previous paragraph in mind, we can revisit the
most intuitive method to minimize convex functions: \emph{gradient descent}. The
gradient descent method consists in following the directions of steepest
descent, which is the direction opposite to the gradient. This leads to an
iteration of the type:
$y_{t+1} = y_t - \eta \cdot \nabla f(y_t)$. In the view
of primal space and dual space, this iteration suddenly looks strange, because
one is summing a primal vector $y_t$ with a dual vector $\nabla f(y_t)$ which
live in different spaces. In some sense, the gradient descent for Lipschitz
convex functions only makes sense in the $\ell_2$ norm, in which $\norm{\cdot} =
\norm{\cdot}_*$ (see the subgradient descent method in \cite{Nesterov2004}).

This motivated the idea of a map $M:\R^d
\rightarrow Q$ connecting the primal and the dual space. The idea in the mirror
descent algorithm is to keep two vectors $(y_t, z_t)$ one in the primal space
and one in the dual space. In each iteration we compute $\nabla f(y_t)$,
obtaining a dual vector and update:
$$z_{t+1} = z_t - \eta \nabla f(y_t) \qquad \qquad y_{t+1} = M(z_{t+1})$$

It is convenient in the analysis to think of this map as the gradient of a
convex function $M = \nabla \omega^*$. In the usual setup, we define the \emph{mirror map}, which is a convex function
$\omega:Q \rightarrow \R$, $\sigma^{-1}$-strongly convex with respect to $\norm{\cdot}$. Let
$\omega^* : \R^d \rightarrow \R$ be the Fenchel-dual $\omega^*(z) = \sup_{y \in
Q} z^\top y - \omega(y)$ which is a $\sigma$-smooth convex function with respect
to $\norm{\cdot}_*$ by Theorem \ref{thm:smoothness_strong_convexity}.

Notice that $\omega^*$ is defined as a maximum over linear functions of $z$
indexed by $y$. The result known as the envelope theorem states that $\nabla
\omega^*(z_0)$ is the gradient of the linear function maximized at $z_0$.
Therefore: $\nabla \omega^*(z_0) = y \in \argmax_y \{ z_0^\top y - \omega(y) \}$.
This in particular implies that $\nabla \omega^*(z) \in Q$ since $\omega(y) =
\infty$ for $y \notin Q$.

Using the definition of $\omega$ and $\omega^*$ we can define the Mirror Descent
iteration as:
$$z_{t+1} = z_t - \eta \nabla f(y_t) \qquad \qquad y_{t+1} = \nabla \omega^*(z_{t+1})$$

\begin{theorem}\label{thm:mirror-descent}
In the setup described above with $D = \max_{z \in Q} D_\omega(z \Vert
z_0)$, $\eta = \epsilon / \sigma \rho^2$ then in $T \geq 2D\sigma \rho^2 /
\epsilon^2$ iterations, it holds that $\frac{1}{T} \sum_t
\nabla f(y_t)^\top (y_t - y) \leq \epsilon, \forall y \in Q$.
\end{theorem}

\begin{proof}
The idea of the proof is to bound the growth of $\omega^*(z_t)$ using smoothness
property of $\omega^*$:
$$\begin{aligned}
\omega^*(z_t) & \leq \omega^*(z_0) + \sum_{t=0}^{T-1} \nabla \omega^*(z_t)^\top
(z_{t+1} - z_t) + \frac{\sigma}{2} \norm{z_{t+1} - z_t}^2_* \\
& = \omega^*(0) -  \sum_{t=0}^{T-1} \eta  \nabla y_t^\top \nabla f(y_t) +
\frac{\sigma}{2} \eta^2 \norm{\nabla f(y_t)}_*^2
\end{aligned}$$
By the Fenchel inequality $\omega^*(z_t) \geq z_t^\top y - \omega(y) = (z_0
- \sum_t \eta \nabla f(y_t))^\top y - \omega(y^*)$ for all $y \in Q$. Combining
with the previous inequality and re-arranging the terms, we get:
$$\eta \sum_t \nabla f(y_t)^\top (y_t - y) \leq \omega(y) + \omega^*(z_0) -
\nabla \omega(y_0)^\top y  + \frac{\sigma}{2} \eta^2 \rho^2 T$$
The gradient of $\omega^*(z_0) = \sup_y z_0^\top y - \omega(y)$ corresponds by
the envelope theorem to $y$ maximizing $z_0^\top y - \omega(y)$. Therefore,
since $y_0 = \nabla \omega^*(z_0)$, $\omega^*(z_0) = z_0^\top y_0 - \omega(y_0)$.
Substituting $\omega^*(z_0)$ in the above expression and using the definition of
Bregman divergence, we get:
$$\eta \sum_t \nabla f(y_t)^\top (y_t - y) \leq
D_\omega(y \Vert y_0) + \frac{\sigma}{2} \eta^2 \rho^2 T$$
Rearranging the terms and using that $ D_\omega(y \Vert y_0) \leq D$, we obtain:
$$\frac{1}{T} \sum_t \nabla f(y_t)^\top (y_t - y) \leq \frac{D}{\eta T} 
+ \frac{\sigma \eta \rho^2}{2} = \sqrt{\frac{2D\sigma
\rho^2}{T}} \text{  for  } \eta = \sqrt{\frac{2D}{T\sigma \rho^2}}$$
So for $T \geq \frac{2 \sigma D \rho^2 }{\epsilon^2}$, $\frac{1}{T} \sum_t
\nabla f(y_t)^\top (y_t - y) \leq \epsilon$.
\end{proof}

\begin{corollary}\label{cor:mirror-descent}
In the conditions of the previous theorem, for $\bar{y}_t =
\frac{1}{T} \sum_{t=1}^T y_t$, $f(\bar{y}_t) - f^* \leq \epsilon$, where $f^* =
\min_{y \in
Q} f(y)$
\end{corollary}

\begin{proof}
Let $y^* = \argmin_{y \in Q} f(y)$. Applying the previous theorem with $y = y^*$
we get:
$$f(\bar{y}_t) - f(y^*) \leq \frac{1}{T} \sum_t f(y_t) - f(y^*) \leq \frac{1}{T}
\sum_t \nabla f(y_t)^\top (y_t - y^*) \leq \epsilon$$
where both inequalities follow from convexity of $f$.
\end{proof}

%% file: lineartime.tex
%!TEX root = main.tex
\section{Nearly linear time deterministic algorithm}\label{sec:linear}

In this section, we present a nearly linear time deterministic algorithm for the
approximate Carath\'{e}odory Problem. Barman's original proof \cite{Bar15} involves
solving the exact Carath\'{e}odory problem, i.e. writing $u = \sum_x x \cdot
\lambda_x$, interpreting $\lambda$ as a probability distribution over $X$,
sampling $k$ points from $X$ according to $\lambda$ and arguing using
concentration bounds (Khintchine inequality to be precise) that the expectation
$\E \norm{u - \frac{1}{k} \sum_{i=1}^k x_i}_p \leq \epsilon$. From an
algorithmic point of view, this requires: (i) solving a linear program to
compute $\lambda$; (ii) using randomization to sample $x_i$. Our main theorem
shows that neither is necessary. There is a linear time deterministic
algorithms that doesn't require a solution $\lambda$ to the exact Carath\'{e}odory
problem.

Our algorithm is based on Mirror Descent. The idea is to formulate the
Carath\'{e}odory problem as an optimization problem. Inspired by early positive
Linear Programming solvers such as the one of Plotkin, Shmoys and Tardos \cite{PST91}, we
convert this problem to a saddle point problem and then solve the dual using
Mirror Descent. Using Mirror Descent to solve the dual guarantees a sparse
primal certificate that would act as the desired convex combination.

Recall that we are given a finite set of points
$X = \{v_1, v_2, \hdots, v_m\} \subseteq \BB_p(1)$ and $u \in \hull(X)$. Our
goal is to produce a sparse convex combination of the points in $X$ that is
$\epsilon$-close to $u$ in the $\ell_p$-norm. Dropping the sparsity constraint
for now, we can formulate this problem as:
\begin{equation}\label{cara-primal}\tag{\textsc{P-Cara}}
\min_{x \in \Delta} \norm{Vx - u}_p
\end{equation}
where $V$ is a $d \times m$ matrix where the columns are the vectors $v_1,
\hdots, v_m$ and $\Delta = \{x \in \R^d \vert \sum_i x_i = 1, x \geq 0\}$ is the unit
simplex in $d$-dimensions. We refer to \ref{cara-primal} as the primal
Carath\'{e}odory problem. This problem can be converted to a saddle point problem by
noting that we can write the $\ell_p$ norm as $\norm{x}_p = \max_{y: \norm{y}_q
= 1} y^\top x$ for $\frac{1}{p} + \frac{1}{q} = 1$. So we can reformulate the
problem as:

\begin{equation}\label{cara-saddle}\tag{\textsc{S-Cara}}
\min_{x \in \Delta} \max_{y \in \BB_q(1) } y^\top(Vx - u)
\end{equation}

Sion's Theorem \cite{Sion58} is a generalization of Von Neumann's minimax theorem
that allows us to swap the order of minimization and maximization for any pair
of compact convex sets. This leads to dual version of the Carath\'{e}odory problem:

\begin{equation}\label{cara-dual}\tag{\textsc{D-Cara}}
\max_{y \in \BB_q(1) } \left(\tilde{f}(y) := \min_{x \in \Delta} y^\top(Vx - u)\right)
\end{equation}

The function $\tilde{f}$ is concave, since it is expressed as a minimum over linear
functions in $y$ parametrized by $x$. Maximizing a concave function is
equivalent to minimizing a convex function. To keep the minimization
terminology, which is more standard in optimization, we write:

\begin{equation}\label{cara-dual-2}\tag{\textsc{D-Cara'}}
- \min_{y \in \BB_q(1) } \left(f(y) := \max_{x \in \Delta} y^\top(u - Vx)\right)
\end{equation}

\paragraph{Sparse solution by solving the dual.} Since $u \in \hull(X)$, there
is a vector $x \in \Delta$ such that $u = Vx$. Hence, the
optimal solution for \ref{cara-primal} is zero and therefore are the solution of
all equivalent formulations. Even though we know the optimal solution, it makes
sense to optimize $f(y)$ since in the process we can obtain an
$\epsilon$-approximation in a few number of iterations. If each iteration
updates only one coordinate, then we are guaranteed to obtain an approximation
with sparsity equal to the number of iterations. As it will become clear in a
second, while the updates of variable $y$ are not sparse, the dual certificate
produced by Mirror Descent will be sparse.

To make this statement precise, consider the gradient of $f$, which can be
obtained by an application of the envelope theorem: $\nabla f(y) = u-Vx$ for $x
\in \argmax_{x \in \Delta} y^\top(u - Vx)$. This problem corresponds to
maximizing a linear function over the simplex, so the optimal solution is a
corner of the simplex. In other words, $\nabla f(y) = u - v_i$ where $i =
\argmax_i [-(V^\top y)_i]$. Finally, we can use the Mirror Descent guarantee in
Theorem \ref{thm:mirror-descent} to bound the norm of the average gradient. We
make this precise in the proof of the following theorem.

\begin{remark}
In fact $V$ does not even have to be explicitly given. All we need is to solve   $i =
\argmax_i [-(V^\top y)_i]$. When $V$ is explicitly given, this can be done in $dn$ time by picking the best vertex. Sometimes, especially in combinatorial optimization, we have a polytope (whose vertices are $V$) represented by its constraints. Our result states that for these alternate formulations, we can still obtain a sparse representation efficiently if we can solve linear optimization problems over it fast. This observation will be important for our appication to submodular minimization.
\end{remark}

\begin{theorem}\label{cara-mirror}
Consider a $(1/\sigma)$-strongly convex function $\omega: \BB_q(1) \rightarrow \R$
with respect to the $\ell_q$-norm,
$D = \max_{y \in \BB_q(1)} D_\omega(y \Vert 0)$ and $T \geq 8D\sigma /
\epsilon^2$. Let $y_1 = 0, \hdots, y_T$ be the $T$ first iterates of the
Mirror Descent algorithm (Theorem \ref{thm:mirror-descent}) with mirror map
$\nabla \omega^*$ minimizing function $f$ in \ref{cara-dual-2}. If 
$\nabla f(y_t) = u - v_{i(t)}$, then
$$\norm{u - \frac{1}{T} \sum_{t=1}^T v_{i(t)}}_p \leq \epsilon.$$
\end{theorem}

\begin{proof}
We consider the space $y \in \BB_q(1)$ equipped with the $\ell_q$ norm. To apply
the Mirror Descent framework, we need first to show that the dual norm (the
$\ell_p$-norm, in this case) of the gradient is bounded. This is easy, since in
the approximate Carath\'{e}odory problem, $v_i \in \BB_p(1)$, so $\norm{\nabla
f(y)}_p = \norm{u - v_i}_p \leq \norm{u}_p + \norm{v_i}_p \leq 2$. So we can
take $\rho = 2$ in Theorem \ref{thm:mirror-descent}.

Since $f(y) = \max_{x \in \Delta} y^\top (u - Vx)$ and $\nabla f(y) =
(u-Vx)$ for $x \in \argmax_{x \in \Delta} y^\top(u-Vx)$, then $f(y) = \nabla
f(y)^\top y$. Also, since $ f(y)$ can be written as $\norm{u - Vx}_\infty$,
clearly $f(y) \geq 0$ for all $y$. Plugging those two facts in the guarantee of
Theorem \ref{thm:mirror-descent}, we get:

$$\epsilon \geq \frac{1}{T} \sum_{t=1}^T \nabla f(y_t)^\top (y_t - y) =
\frac{1}{T} \sum_{t=1}^T [f(y_t) - \nabla f(y_t)^\top y] \geq
\left[-\frac{1}{T} \sum_{t=1}^T   \nabla f(y_t) \right]^\top y, \forall y \in
\BB_q(1)$$

Taking the maximum over all $y \in \BB_q(1)$ we get:
$$ \norm{u - \frac{1}{T} \sum_{t=1}^T v_{i(t)}}_p = \norm{\frac{1}{T}
\sum_{t=1}^T \nabla f(y_t) }_p = \max_{y \in \BB_q(1) } \left[-\frac{1}{T}
\sum_{t=1}^T   \nabla f(y_t) \right]^\top y \leq \epsilon $$
\end{proof}

To complete the picture, we need to provide a function a $(1/\sigma)$-strongly
convex function $\omega : \BB_q(1) \rightarrow \R$ with a small value of $\sigma
\cdot \max_{y \in \BB_q(1)} D_\omega(y \Vert 0)$.

\begin{proposition}\label{prop:mirror-map}
For $1<q\leq 2$, the function
$\omega: \BB_q(1) \rightarrow \R$, $\omega(y) = \frac{1}{2} \norm{y}_q^2$
is $(q-1)$-strongly convex with respect to the
$\ell_q$ norm and $\max_{y \in \BB_q(1)} D_\omega(y \Vert 0) = \frac{1}{2}$.
\end{proposition}

\begin{proof}
We want to bound $\omega(y) - \omega(x) - g^\top (y-x)$ for all $g \in \partial
\omega(x)$. For all $x$ in the interior of the ball $\BB_q(1)$ there is a unique
subgradient which we represent by $\nabla \omega(x)$. In the border of
$\BB_q(1)$, however, there are multiple subgradients. First we claim that we
need only to bound $\omega(y) - \omega(x) - \nabla \omega(x)^\top (y-x)$ where
$\nabla \omega(x)$ denotes the gradient of the function $\frac{1}{2}
\norm{y}_q^2$. In order to see that, notice that if $g$ is a subgradient in a
point $x$ and $y \in \BB_q(1)$ then:
$$\omega(x+t(y-x)) - \omega(x) - g^\top (y-x) \geq 0$$
by the definition of subgradient. Dividing the expression by $t$ and taking the
limit when $t \rightarrow 0+$, we get: $\nabla \omega(x)^\top (y-x) \geq g^\top
(y-x)$, so in particular: $\omega(y) - \omega(x) - g^\top (y-x) \geq \omega(y)
- \omega(x) - \nabla \omega(x)^\top (y-x) $. 

This observation allows us to bound the strong convexity parameter of
$\omega$ by looking at the $\norm{\cdot}_q$-eigenvalues of the Hessian of $\omega$.  
In particular, we will
show that for all $w \in \R^d$, $w^\top \nabla^2 \omega(y) w \geq (q-1)
\norm{w}_q^2$.

To make the notation simpler, we define $\pow: \R^d \times \R \rightarrow \R^d$
as $\pow(y,p) = (\abs{y_i}^p \cdot \sgn{y_i})_i$. This allows us to represent
$\nabla \norm{y}_q$ in a succinct form: since $$\partial_i \norm{y}_q =
\frac{1}{q} (\norm{y}_q^q)^{\frac{1}{q}-1} x_i^q q \cdot \sgn{x_i} = \norm{y}_q^{1-q}
x_i^{q-1} \sgn{y_i}$$
so we can write $\nabla \norm{y}_q = \norm{y}_q^{1-q} \cdot \pow(y,q-1)$.
Therefore:
$$\nabla \omega(y) = \nabla \left[ \frac{1}{2}\norm{y}_q^2 \right] =
\norm{y}_q \cdot \nabla \norm{y}_q = \norm{y}_q^{2-q} \cdot \pow(y,q-1)$$
Now, to compute the Hessian, we have:
$$\nabla^2 \omega(y) = (2-q) \norm{y}_q^{2-2q} \cdot \pow(y,q-1)
\pow(y,q-1)^\top + (q-1) \norm{y}^{2-q}_q \diag(\abs{y_i}^{q-2})$$
where $\diag(\abs{y_i}^{q-2})$ is the diagonal matrix with $x_i^{q-2}$ in the
diagonal. Using the fact that $1 < q \leq 2$, we can write:
$$\begin{aligned}
w^\top \nabla^2 \omega(y) w & = (2-q) \cdot \norm{y}_q^{2-q}
[\pow(y,q-1)^\top w]^2 + (q-1) \cdot \norm{y}_q^{2-q} \sum_i \abs{y}_i^{q-2} w_i^2 \\
& \geq (q-1) \left( \sum_i \abs{y}_i^q \right)^{\frac{2-q}{q}} \cdot 
\left( \sum_i \abs{y}_i^{q-2}  w_i^2\right) \\ & = (q-1) \left[\left( \sum_i
\abs{y_i}^{\frac{q(2-q)}{2} \cdot \frac{2}{2-q}} \right)^{\frac{2-q}{2}} \cdot
\left( \sum_i (\abs{y_i}^{\frac{q(q-2)}{2}}  w_i^q)^{\frac{2}{q}}
\right)^{\frac{q}{2}} \right]^{\frac{2}{q}}
\end{aligned}$$

The last equality is a convoluted re-writing of the previous expression, but
allows us to apply H\"{o}lder's inequality. Recall that H\"{o}lder's inequality states
that $\norm{z_1}_a \cdot \norm{z_2}_b \geq z_1^\top z_2$ whenever $\frac{1}{a} +
\frac{1}{b} = 1$. Applying this inequality with $a = \frac{2}{2-q}$ and $b =
\frac{2}{q}$, we get:
$$w^\top \nabla^2 \omega(y) w \geq (q-1) \cdot \left( \sum_i
\abs{y_i}^{\frac{q(2-q)}{2}} \cdot \abs{y_i}^{\frac{q(q-2)}{2}} w_i^q
\right)^{\frac{2}{q}} = (q-1) \cdot \left( \sum_i w_i^q \right)^{\frac{2}{q}} =
(q-1) \cdot \norm{w}_q^2$$
\end{proof}

Finally, we need to show how to compute the Fenchel dual $\omega^*$ and the 
mirror map $\nabla \omega^*$ efficiently:

\begin{proposition}\label{prop:fenchel_dual_computation}
The Fenchel dual of the function $\omega$ defined in Proposition
\ref{prop:mirror-map} can be computed explicitly:
$$\omega^*(z) =
        \begin{cases} 
          \frac{1}{2} \norm{z}_p^2 & \textnormal{if } \norm{z}_p \leq 1 \\
          \norm{z}_p - \frac{1}{2} & \textnormal{if } \norm{z}_p > 1
        \end{cases}$$
Also, $\nabla \omega^*(z) = \phi(z) \cdot \min(1, \norm{z}_p)$ where $\phi(z)$ is
a vector with $\ell_q$-norm $1$ such that $z^\top \phi(z) = \norm{z}_p$. This
function can be explicitly computed as:
$\phi(z)_i = \sgn{z_i} \cdot \abs{z_i}^{p-1} / \norm{z}_p^{p-1}$.
\end{proposition}

\begin{proof}
By the definition of Fenchel duality:
$$\omega^*(z) = \max_{y \in \BB_q(1)} z^\top y - \frac{1}{2}\norm{y}_q^2 =
\max_{0 \leq \lambda \leq 1} \left[ \max_{\hat{y}; \norm{\hat{y}}_q = 1} \lambda
z^\top \hat{y} - \frac{1}{2}\lambda^2 \right] = \max_{0 \leq \lambda \leq 1}
\lambda \norm{z}_p - \frac{1}{2} \lambda^2$$
where the second equality follows from writting $y = \lambda \hat{y}$ for
$0 \leq \lambda \leq 1$ and $\norm{\hat{y}}_p = 1$. The optimal value of
$\hat{y}$ is $\phi(z)$. The expression $\lambda \norm{z}_p - \frac{1}{2}
\lambda^2$ is maximized at $\lambda = \norm{z}_p$. Since $\lambda$ is restricted
to lie between $0$ and $1$, the optimal $\lambda$ must be $\min(1, \norm{z}_p)$.

If $\norm{z}_p \leq 1$, $\lambda = \norm{z}_p$ and $\omega^*(z) = \frac{1}{2}
\norm{z}_p^2$. If $\norm{z}_p > 1$, then $\lambda = 1$ and $\omega^*(z) =
\norm{z}_p - \frac{1}{2}$.

By the envelope theorem, $\nabla \omega^*(z) = \hat{y} \cdot \lambda =  \phi(z)
\cdot \min(1, \norm{z}_p)$. Now, it simple to check that $\phi$ has the desired
properties:
$$\norm{\phi(z)}_q^q = \sum_i \abs{z_i}^{q(p-1)} / \norm{z}_p^{q(p-1)} = \sum_i
\abs{z_i}^p / \norm{z}_p^p = 1$$
$$z^\top \phi(z) = \sum_i z_i^p / \norm{z}_p^{p-1} = \norm{z}_p$$
\end{proof}

Combining the previous results, we obtain:

\begin{theorem}\label{thm:cara-mirror-main}
Given $n$ points $v_1, \hdots, v_n \in \BB_p(1) \subseteq \R^d$ with $p \geq 2$
and  $u \in \hull\{v_1, \hdots, v_n\}$, there is a deterministic algorithm
of running time $O(nd \cdot p/\epsilon^2)$ that a outputs a multiset 
$v_{i(1)}, \hdots, v_{i(k)}$ for $k = 4(p-1)/\epsilon^2$ 
such that $u' = \frac{1}{k} \sum_{t=1}^k v_{i(t)}$ and $\norm{u'-u}_p \leq
\epsilon$.
\end{theorem}

\begin{proof}
The number of iterations of Mirror Descent (and consequently the sparsity bound)
$T = 4p / \epsilon^2$ can be obtained by substituting $D = 1/2$ and
$\sigma^{-1} = (q-1)^{-1} = p-1$ from Proposition \ref{prop:mirror-map}
in Theorem \ref{cara-mirror}.

For the running time, notice that the time per iteration is dominated by the
computation of the subgradient of $f$. The most expensive step is the
computation of $V^\top y$ which takes $dn$ operations, which is the size of
matrix $V$.

\end{proof}

\subsection{Improved bound when $u$ is far from the boundary}
\begin{theorem}\label{thm:improved_bd}
Let $P$ be a polytope contained inside the unit $\ell_{p}$ ball,
and a point $u\in P$. If $\BB_{p}(r)\subseteq P$, then there exists
$x\in2(1-\epsilon/r)\cdot\Delta$ supported at $k=O\left(\frac{p}{r^{2}}\cdot\log\frac{r}{\epsilon}\right)$ coordinates such
that $\left\Vert \sum_{i\in\textnormal{supp}(x)} x_{i}v_{i}-u\right\Vert _{p}\leq\epsilon$.
\end{theorem}
\begin{proof}
Let $\lit{ApproxCara}(u)$ represent the convex combination $x$ of vertices
of $P$ returned by the algorithm from Theorem~\ref{thm:cara-mirror-main}, when receiving as
input the vertices of $P$ and the point $u$, and solving for precision
$r/2$. Then let $e_{0}=u$, $x_{i}=\lit{ApproxCara}(e_{i-1})$, $e_{i}=2(e_{i-1}-Vx_{i})$,
for $i\in\left\{ 1,\dots,\beta\right\}$ where $\beta=\log(r/\epsilon)$.
Note that $\Vert e_{i}\Vert_{p}=2\Vert e_{i-1}-Vx_{i}\Vert_{p}\leq2\cdot\frac{r}{2}\leq r$,
hence $e_{i}\in P$, so the input to $\lit{ApproxCara}$ is always well
defined. Let $\overline{x}=\sum_{i=1}^{\beta}\frac{1}{2^{i-1}}\cdot x_{i}\in\left(\sum_{i=1}^{\beta}2^{-(i-1)}\right)\Delta=\frac{2^{-\beta}-1}{2^{-1}-1}\cdot\Delta=2(1-\epsilon/r)\Delta$.

Let us bound the error when approximating $u$ with $V\overline{x}$:
\begin{align*}
\Vert V\overline{x}-u\Vert_{p} & =\left\Vert V\left(\sum_{i=1}^{\beta}\frac{1}{2^{i-1}}\cdot x_{i}\right)-u\right\Vert _{p}
  =\left\Vert \sum_{i=1}^{\beta}\frac{1}{2^{i-1}}\cdot Vx_{i}-e_{0}\right\Vert _{p}
  =\left\Vert \sum_{i=2}^{\beta}\frac{1}{2^{i-1}}Vx_{i}+Vx_{1}-e_{0}\right\Vert _{p}\\
 & =\left\Vert \sum_{i=2}^{\beta}\frac{1}{2^{i-1}}Vx_{i}-\frac{1}{2}e_{1}\right\Vert _{p}
  =\frac{1}{2}\left\Vert \sum_{i=2}^{\beta}\frac{1}{2^{i-2}}Vx_{i}-e_{1}\right\Vert _{p}
  =\dots\\
 & =\frac{1}{2^{\beta-1}}\left\Vert \sum_{i=\beta}^{\beta}\frac{1}{2^{i-\beta}}Vx_{i}-e_{\beta-1}\right\Vert _{p}
  =\frac{1}{2^{\beta-1}}\left\Vert Vx_{\beta}-e_{\beta-1}\right\Vert _{p}
  \leq\frac{1}{2^{\beta-1}}\cdot\frac{r}{2}\\
 & =r/2^{\beta} =\epsilon
\end{align*}

Each of the $\beta=\log(r/\epsilon)$ iterations requires a call to
$\lit{ApproxCara}$ for precision $r/2$, which produces a solution with
sparsity $O(p/r^{2})$. Hence $\overline{x}$ will have $O\left(\frac{p}{r^{2}}\log\frac{r}{\epsilon}\right)$ nonzero
coordinates.
\end{proof}

\begin{corollary}
If $u\in P$ satisfies $\BB_p(u,r)\subseteq P \subseteq \BB_p(u,1)$, $r \geq 2\epsilon$, then there exists $x\in \Delta$ supported at $k=O\left(\frac{p}{r^{2}}\cdot\log\frac{r}{\epsilon}\right)$ coordinates such that $\left\Vert \sum_{i\in\textnormal{supp}(x)} x_{i}v_{i}-u\right\Vert _{p}\leq\epsilon$.
\end{corollary}
\begin{proof}
Let $v_i' = v_i - u$ for all $i$. This corresponds to translating $P$ such that $u$ is placed at the origin. By the triangle inequality, this at most doubles the radius of the origin-centered $\ell_p$ ball circumscribing the polytope. Applying Theorem~\ref{thm:improved_bd} we obtain a vector $x \in 2(1-r/\epsilon)\Delta$ such that $\left\Vert \sum_{i\in\textnormal{supp}(x)} x_i v_i' \right\Vert_p \leq \epsilon$. Let $x_i' = x_i / \Vert x \Vert_1 \in \Delta$. This satisfies $\left\Vert \sum_{i\in\textnormal{supp}(x)} x_i' v_i' \right\Vert_p \leq \epsilon/\Vert x\Vert_1$. Hence $\left\Vert \sum_{i\in\textnormal{supp}(x)} x_i' v_i' \right\Vert_p  =\left\Vert \sum_{i\in\textnormal{supp}(x)} x_i' (v_i-u) \right\Vert_p  = \left\Vert \sum_{i\in\textnormal{supp}(x)} x_i' v_i - u \right\Vert_p \leq \frac{\epsilon}{2(1-\epsilon/r)} \leq \epsilon$.
\end{proof}

%% file: lower.tex
%!TEX root = main.tex
\section{Lower bound}

We showed that if $V$ is a $d \times n$ matrix whose columns are contained in the unit $\ell_p$ ball,
then for any $x \in \Delta_n$ there is $\tilde{x} \in
\Delta_n$ with $\ssupp{\tilde{x}} \leq O(p / \epsilon^2)$ such that $\norm{Ax -
Ax'}_p \leq \epsilon$, where $\supp(x) = \left\{i | x_i \neq 0  \right\}$.

In this section we argue that no dimension independent bound better then $O(p /
\epsilon^2)$ is possible. This shows that the sparsity bound in Theorem
\ref{thm:cara-mirror-main} is tight and improves Barman's
$\Omega(1/\epsilon^{p(p-1)})$ lower bound \cite{Bar15}. Formally, we show
that:

\begin{theorem} 
There exists a constant $K$ such that for every $p \geq 2$ and $n\geq n_0(p)$, there exists $n \times n$ matrix $V$ with columns of unit $\ell_p$ norm, and a point $u=Vx$, $x\in\Delta_n$, such that for all
$\tilde{x} \in \Delta$ with sparsity $\ssupp{\tilde{x}}\leq K p / \epsilon^2$, one has that $\norm{ V\tilde{x} - u}_p \geq 2\epsilon > \epsilon$.
\end{theorem}

In other words, even thought $u$ is a convex combination of columns of $V$,
every $(Kp / \epsilon^2)$-sparse convex combination of columns of $V$ has distance
at least $\epsilon$ from $u$ in the $\ell_p$-norm. We first present a simple and
constructive $\Omega(1/\epsilon^2)$ lower bound and later a tight $\Omega(p /
\epsilon^2)$ lower bound based on the probabilistic method.

\subsection{A simple lower bound $\Omega(1/\epsilon^{2})$}

The lower bound relies on Sylvester's construction of Hadamard matrices, which
are defined for all values of $n$ that are powers of $2$.

Sylvester's recursive construction follows by defining $H_1 =
[1]$, and for every $n$ that is a power of $2$:
$$H_{2n} =\begin{bmatrix} H_n & H_n \\ H_n & -H_n\end{bmatrix}$$

\begin{proposition}
The Sylvester matrix $H_n$ defined as above is Hadamard. In other words,  $H_{ij}=\pm1$ for all $i,j$ and
$H^{T}H=nI$ (its columns are mutually orthogonal).
\end{proposition}

Now we consider the polytope $P$ formed by the convex hull of the
normalized columns of $H$. One can easily check that for the construction above
the uniform combination of columns is $H \cdot \vec{1}/n = e_1$ where $\vec{1}$ is the
vector of all $1$'s and $e_i$ is the unit vector in the direction of the the
$i$-th component. We show that $e_1$ is at distance greater than $\epsilon$ from the convex hull
of any $o(1/\epsilon^{2})$ columns of $H$.

\begin{theorem}
Let $H_n$ be a $n\times n$ Sylvester matrix, and let $P$ be the convex
hull of the columns of $\tilde{H} := H/n^{1/p}$. Let $u = \tilde{H}\cdot  \vec{1}/n =  e_1 / n^{1/p}\in P$. Then any
$x \in \Delta_n$ satisfying $\left\Vert \tilde{H}x - u \right\Vert_p \leq \epsilon$ has sparsity $\ssupp{x} \geq \min(1/\epsilon^2, n)$.
\end{theorem}

\begin{proof}
Let $x\in\Delta_{n}$ be $k$-sparse, i.e. $\ssupp{x} = k$, such that
$\norm{\tilde{H}x - u}_{p} \leq \epsilon$. We would like to lower bound the
sparsity $k$ in terms of $\epsilon$ and $p$.

We will use two main ingredients in the proof: the first is the \emph{power mean
  inequality} which states that for any vector $x \in \R^n$, $\left( \frac{1}{n}
  \sum_i x_i^t\right)^{1/t}$ is non-decreasing in $t$. In particular, this
  implies that $\norm{x}_t \cdot n^{-1/t}$ is non-decreasing\footnote{This can
    be seen by computing the derivative of $M_t(x) = \left( \frac{1}{n}
        \sum_i x_i^t\right)^{1/t}$ with respect to $t$ and showing it is
      non-positive.}. 
    The second fact we
  will use is that for every vector $\norm{x}_1^2 \leq \norm{x}_2^2 \cdot
\ssupp{x}$. This follows from the Cauchy-Schwarz inequality:
  $\norm{x}_1^2 = \left( \sum_{i\in\supp(x)} x_i \cdot 1 \right)^2 \leq \left(\sum_{i\in\supp(x)}
  x_i^2\right) \cdot \left( \sum_{i\in\supp(x)} 1 \right) = \norm{x}_2^2 \cdot \ssupp{x}$.
Combining both results give us a bound involving the $2$-norm of the error:
$$\epsilon \geq \norm{\tilde{H}x - u}_p = \frac{1}{n^{1/p}} \cdot \norm{Hx -
e_1}_p \geq  \frac{1}{n^{1/2}} \cdot \norm{Hx -
    e_1}_2$$
    where the last step follows from the power-mean inequality. Squaring both
    sides, we get:
    $$\epsilon^2 \geq \frac{1}{n} (Hx - e_1)^\top (Hx - e_1) = \frac{1}{n}
    \left[x^\top H^\top H x - 2
    e_1^\top Hx + 1\right] = \norm{x}_2^2 -\frac{1}{n} \geq \frac{\norm{x}_1^2}{ \ssupp{x}} -
    \frac{1}{n} = \frac{1}{k} - \frac{1}{n}$$
    We used the fact that $e_1^\top Hx = \norm{x}_1 = 1$ since the top row of $H$ consists of only
$1$'s. 
Hence $k\geq \left(\epsilon^2 + 1/n \right)^{-1} \geq 1 / \max\left(   \epsilon^2, 1/n     \right) = \min\left(  1/\epsilon^2, n \right)$.
\end{proof}

\subsection{Tight lower bound $\Omega(p/\epsilon^{2})$}

We now establish a tight lower bound via a probabilistic existence argument
inspired by the construction of Klein and Young \cite{KleinY15}. The example
used to exhibit the lower bound is very simple. The proof of its validity,
however, is quite involved and requires a careful probability analysis.
We first give an overview before delving into the details. \\

\paragraph{Overview.} 
Recall the formulation \ref{cara-saddle} of the Carath\'{e}odory problem as a saddle
point problem described in Section \ref{sec:linear}. If we translate all points such
that $u = 0$, then we can write the problem as:
$$\min_{x \in \Delta} \max_{y \in \BB_q(1)} y^\top V x$$
which can be seen as a game between a player controlling $x$ and $y$. The
approximate Carath\'{e}odory theorem states that (under the conditions from Theorem
\ref{thm:cara-mirror-main}), if the value of the game is zero, then the
$x$-player has a $k$-sparse strategy that guarantees that the value of the game
is at most $\epsilon$ for $k = O(p/\epsilon^2)$.

For the lower bound, our goal is to construct an instance of this game with
value $v$ such that for all $k$-sparse strategies of the $x$-player with $k < C p /
\epsilon^2$, the $y$-player can force the game to have a value strictly
larger than  $v + \epsilon$.

\emph{Probabilistic Construction:} We  define the matrix $V$ such that
$V = n^{-1/p} \cdot A$ where $A$ is an $n \times n$ matrix with random $\pm 1$
entries, i.e. each entry of $A$ is chosen at random from $\{-1, +1\}$ independently
with probability
$1/2$. Note the the $\ell_p$ norm of the columns of $A$ is equal to $1$, so we
are in the conditions of Theorem \ref{thm:cara-mirror-main}.
Then we will show that the following events happen with high probability:

\begin{enumerate}
\item The center of the polytope defined by the columns of $V$ is $\epsilon$-close
to zero, i.e., $\norm{V \cdot \one/n}_p \leq \epsilon$.
\item For each set $S$ of $k$-coordinate, if $x$ is restricted to only use those
coordinates, the $y$-player can force the value of
the game to be at least $2\epsilon$. We prove so by exhibiting a strategy for
the $y$-player such that $y^\top V$ is at least $2\epsilon$ for all coordinates in $S$.
\end{enumerate}

After we bound the probabilities of the previous events, the result follows by
taking the union bound over all the $\binom{n}{k}$ possible subsets $S$ of cardinality
$k$. This will show that with nonzero probability a matrix will be constructed, for which
the $y$ player will always be able  force $y^T V x \geq 2\epsilon$, no matter what
$o(p/\epsilon^2)$-sparse strategy player $x$ chooses.

\paragraph{Bounding probabilities.}

All the lemmas and theorems from this of this section are in the conditions of
the previous paragraphs: $A$ and $V$ are the random matrices previously
defined, and $S$ is a fixed subset of $x$-coordinates of size $k$.

\begin{lemma}\sloppy If the $x$-player plays the uniform strategy, then
$\E\left[\norm{ V \cdot \one/n}_p\right] \leq \sqrt{p/n}$, and  $\P\left[\norm{V \cdot \one/n}_p \geq
\epsilon\right] \leq \sqrt{\frac{p}{n \epsilon^2}}$ for 
$n \geq p/\epsilon^2$.
\end{lemma}

\begin{proof} The bound on the expectation follows from Khintchine's inequality, 
which states that for any given vectors $u_1, \hdots, u_m
\in \R^n$ and iid uniform $\{-1,+1\}$-variables $r_i$
$$\E\norm{\sum_i r_i u_i}_p \leq \sqrt{p} \cdot \left(\sum_i \norm{u_i}_p^2
\right)^{1/2} $$
We refer to \cite{wolfflectures} or \cite{Bar15} for a proof. Let $v_i$ be the
columns of $V$. Since they are iid uniform, $v_i$ has the same distribution of
$r_i v_i$ for some uniform $\{-1,+1\}$-variable $r_i$. So:
$$\E \norm{V \left( \frac{1}{n} \one\right)}_p = \E \norm{\sum_i r_i
\frac{v_i}{n}}_p \leq \sqrt{p} \cdot \left( n \cdot \frac{1}{n^2} \right)^{1/2} =
\sqrt{\frac{p}{n}}$$
The second part of the lemma is direct from Markov's inequality.
\end{proof}

For the second part, consider an $x$-player that is restricted to only use
coordinates from $S$. We want to show that the $y$-player has a strategy that
that would make all columns in $S$ have a high value in $y^\top V$. The idea is
that since $n$ is large and $k$ is small (in fact, a constant independent of
$n$) there should be rows that are very skewed
(i.e. have a lot more $+1$'s than $-1$'s in the $S$ columns). If $y$ plays
a strategy that only uses such rows, then he can force $x$ to have high value.

We call a row of $A$ is \emph{good} for the $y$-player if it  has more than
$(1/2+\epsilon)k$ $1$'s in the $S$-coordinates. Next we show that with high
probability there is a large enough number of good rows available for the $y$-player. For
this result, we need to lower bound the probability in the tail of the
binomial distribution, which requires a tight anti-concentration inequality.

Anti-concentration can be derived by carefully plugging the moment
generating function into the Paley-Zygmund~\cite{PaleyZ32} inequality, or by sharply estimating 
a sum of terms involving binomial coefficients. For further information we refer the reader
to Tao's book \cite{taotopics}.

\begin{lemma}[Chernoff bound]\label{lemma:bin-tail-bound}
  If $0 < \rho \leq 1/2$, $X_i$ are iid $\{0,1\}$-random variables with $\P[X_i
  = 1] = \rho$ and $\sigma^2 = k \rho (1-\rho)$ is the variance of $X = \sum_{i=1}^k
  X_i$, then there exist constants $C,c > 0$ such that for all $\lambda \leq c
  \sigma$,
  $$\P\left[\abs{{\textstyle \sum_{i=1}^k} X_i - \rho\cdot k} \leq \lambda \sigma\right] \geq 
  c \exp\left(-C \lambda^2 \right)$$
  \end{lemma}

\begin{lemma}[Anti-concentration for the binomial distribution]\label{lemma:bin-tail-bound}
  In the same conditions as in the previous lemma, there exist
  constants $\tilde{C},\tilde{c} > 0$ such that for all $\lambda \leq c
  \sigma$,
  $$\P\left[\abs{{\textstyle \sum_{i=1}^k} X_i - \rho\cdot k} \geq \lambda \sigma\right] \geq 
  \tilde{c} \exp\left(-\tilde{C} \lambda^2 \right)$$
  \end{lemma}

\begin{lemma}\label{lemma:good-rows}\sloppy
There are $r = \Omega(n \exp(-O(k \epsilon^2))$ good rows with probability
at least $1-\exp(-\Omega(n \exp(-O(k \epsilon^2))))$.
\end{lemma}

\begin{proof}
Applying Lemma \ref{lemma:bin-tail-bound} with $\rho = 1/2$ and $\lambda = 4 \epsilon
\sigma$ we obtain that the probability that a row is at least
$k\left(\frac{1}{2}+\epsilon\right)k$ $+1$'s (or $-1$'s, due to symmetry) is at least $\exp(-O(k \epsilon^2))$.
So in expectation, there are $n \exp(-O(k \epsilon^2))$ good rows, so the
result follows by applying the Chernoff bounds with $\rho = \exp(-O(k \epsilon^2))$
and $\lambda = c \sigma$.
\end{proof}

With high probability there will be at least
$r = \Omega(n \exp(-O(k \epsilon^2))$ good rows for the $y$-player
to play. We want to argue that if the $y$-player
can find $r$ good rows, then he can play $y_i = r^{-1/q}$ for each good row $i$,
and $0$ otherwise, and he will leave an $x$-player restricted to choosing only
columns from a subset $S$ without any good option to play.

\begin{lemma}\label{lemma:many-plus-ones}
Let $S$ be a fixed subset of columns of $A$, and let $A_S$ be the matrix whose columns
are the columns of $A$ that belong to $S$. Conditioning on $A_S$ having $r$ good rows, 
with probability at least $1-k\exp\left(-\Omega(r \epsilon^2)\right)$,
every column in $S$ contains at least $r(1/2+\epsilon/2)$  $+1$'s in the $r$ good rows.
\end{lemma}

\begin{proof}
Sample matrix $A$ according to the following procedure: in the first phase, sample
each entry of $A$ uniformly and independently from $\{-1,+1\}$. In the second
phase, for each row, shuffle the entries in $S$ (i.e. for each row, sample a
random permutation of $S$ and apply to the entries corresponding to those
columns). In the first phase we can decide which
rows are good, call those $R$. Conditioning on the first phase, and fixing a
column $j \in S$, the entries $A_{ij}$ for $i \in R$ are independent and uniform from
$\{-1, +1\}$;  the probability of $A_{ij}$ being $1$ is at least $\frac{1}{2} +
\epsilon$, since this entry is a random entry from a good row sampled in the
first phase. The result follows by applying the Chernoff bound with $\rho =
\frac{1}{2} + \epsilon$ and $\lambda = (\sigma \epsilon) / (\rho \cdot (1-\rho))$.
\end{proof}

Now, we combine all the events discussed so far using the union bound:

\begin{lemma}\label{lemma:prob-lemma}
Fix $\epsilon$ and $k$. For sufficiently large $n$, there is
a matrix $A$ such that $V = n^{-1/p} \cdot A$ satisfies $\norm{V \cdot \vec{1}/n}_p
\leq \epsilon$, and for every subset $S$ of $k$ rows, there is a subset $R$ of
$r = \Omega(n \exp(-O(k \epsilon^2)))$ rows such that for all $i \in S$,
$\sum_{j \in R} A_{ij} \geq \epsilon r$.
\end{lemma}

\begin{proof}
The proof follows from the probabilistic method.
For each subset $S$, with probability at least $1-\exp(-\Omega(n \exp(-O(k
\epsilon^2))))$ there are $r$ good rows (Lemma \ref{lemma:good-rows}) 
and with probability $1-k\exp\left(-\Omega(r \epsilon^2)\right) 
=  1-k\exp\left(-\Omega(n \epsilon^2 \exp(-O(k \epsilon^2)))\right)$
there are at least $(\frac{1}{2}+\epsilon)r$
many $+1$s in each column corresponding to the $r$ rows (Lemma
\ref{lemma:many-plus-ones}), causing $\sum_{j \in R}
A_{ij} \geq \epsilon r$. The probability that both events occur can be bounded
by $1- O\left(  k\exp\left(-\Omega(n
\epsilon^2 \exp(-O(k \epsilon^2)))\right) \right)$.
Applying the union bound over all ${n \choose k}$ subsets $S$, we get:
$$1-{n \choose k} O\left(  k\exp\left(-\Omega(n
\epsilon^2 \exp(-O(k \epsilon^2)))\right) \right) \geq 1-\exp\left( k \log n - O(
\epsilon^2 n \exp(-O(k\epsilon^2) )) \right) $$
which goes to one as $n \rightarrow\infty$ for any fixed $k$ and $\epsilon$.
Also, as $n \rightarrow \infty$ the probability that  $\norm{V(\frac{1}{n})
\one}_p \leq \epsilon$ also goes to $1$.
\end{proof}

\begin{theorem}[Carath\'{e}odory lower bound] There is a matrix $V$ whose columns
have unit $\ell_p$ norm such that $\norm{V \cdot \vec{1}/n}_p \leq
\epsilon$, and for every $x \in \Delta$, $\ssupp{x} \leq k = O(p / \epsilon^2)$,
$\norm{Vx}_p \geq 2 \epsilon$.
\end{theorem}

\begin{proof}
Let $V$ be the matrix obtained in Lemma \ref{lemma:prob-lemma}. From there,
we have that $\norm{V \cdot \vec{1}/n}_p \leq \epsilon$. Now, fix any $x \in
\Delta$ with $\ssupp{x} \leq k$, and let $S$ be the set of columns corresponding to the
support of $x$. Let also $R$ be the set of rows for which $\sum_{j \in R} A_{ij}
\geq \epsilon r$  for all $i \in S$. Now, define $y \in \BB_q(1)$ such that $y_i =
r^{1/q}$ uf $i \in R$, and $y_i = 0$ otherwise:
$$\norm{Vx}_p \geq y^\top V x = n^{-1/p} \cdot (y^\top A) x \geq \frac{r
\epsilon}{r^{1/q} \cdot n^{1/p}} = \epsilon \left(\frac{r}{n}\right)^{1/p}$$
We want to choose the parameters such that $\left(\frac{r}{n}\right)^{1/p} \geq
2$. Substituting $r = \Omega(n \exp(-O(k \epsilon^2)))$:
$$\left(\frac{r}{n}\right)^{1/p} = \exp\left(-O\left(\frac{k
\epsilon^2}{p}\right)\right)$$
If $k \leq C \cdot \frac{p}{\epsilon^2}$ for a suitable constant $C$, we get
$\norm{Vx}_p \geq 2 \epsilon$. 
\end{proof}

%% file: app.tex
%!TEX root = main.tex
\providecommand{\tabularnewline}{\\}

\section{Applications}

The approach presented in the previous sections can be easily generalized
or directly applied to a series of applications. Here we identify
three representative applications to illustrate the usefulness of
our approach. We note that there are many other possible applications
in combinatorial optimization, game theory and machine learning, where
a convex combination is often maintained as a subroutine of the algorithm.

\subsection{Warm-up: fast rounding in polytopes with linear optimization oracles}

The most direct application of our approach is to efficiently round
a point in a polytope whenever it admits a fast linear optimization
oracle. An natural such instance is given by the matroid polytope.
We denote a $n$-element matroid by $\mathcal{M}$ and its rank by
$r$.

\begin{proposition}
There is an algorithm which, given a fractional point $x^{*}$ contained
inside the base polytope of a matroid $\mathcal{M}$, and a norm parameter
$p\geq2$, produces a distribution $\mathcal{D}$ over matroid bases
supported on $O\left(\frac{p\cdot r{}^{2/p}}{\epsilon^{2}}\right)$
points, such that $\left\Vert \mathbb{E}_{x\sim\mathcal{D}}\left[x\right]-x^{*}\right\Vert _{p}\leq\epsilon$.
Furthermore the algorithm requires $O\left(nr^{2/p}p/\epsilon^{2}\right)$
calls to $\mathcal{M}$'s independence oracle.
\end{proposition}
\begin{proof}
The result follows from applying Theorem~\ref{cara-mirror} for $x^{*}$ in the convex
hull of the characteristic vectors for matroid bases. Note that each
of these vectors has sparsity $r$ so their $p$ norm is precisely
$r^{1/p}$. Hence we have the desired sparsity for the support of
$\mathcal{D}$. Each iteration requires maximizing a linear function
over the bases of the polytope, which can be done using the standard
greedy algorithm, and requires $O(n)$ calls to the independence oracle.
\end{proof}

Of course, there are other nice polytopes where the existence of an
efficient linear optimization oracle offers advantages. To this aspect,
we mention the $s$-$t$-flow polytope (i.e. the convex hull of all $s$-$t$
paths), whose oracle is implemented with a single shortest path computation.
This enables us to speed up the path stripping subroutine in the Raghavan-Thompson
randomized rounding algorithm for approximating minimum congestion
integral multicommodity flows~\cite{RaghavanT91}. As described in~\cite{RaghavanT91} the
algorithm takes $O(m^{2})$, which can be improved to near linear
time by carefully using link-cut trees~\cite{KangP15}.
By contrast, approximate Carath\'{e}odory provides a lightweight algorithm
for producing an approximate decomposition into integral paths, without
the need of complicated data structures.

\begin{proposition}
There is an algorithm which, given a fractional $s$-$t$-flow $f^{*}$
routing one unit of demand in $G$, and a norm parameter $p\geq2$,
produces a distribution $\mathcal{D}$ over $s-t$-paths supported
on $O\left(\frac{p\cdot n^{2/p}}{\epsilon^{2}}\right)$ points,
such that $\left\Vert \mathbb{E}_{f\sim\mathcal{D}}\left[f\right]-f^{*}\right\Vert _{p}\leq\epsilon$.
Furthermore the algorithm requires $O\left(\frac{p\cdot n{}^{2/p}}{\epsilon^{2}}\right)$shortest
path computations.
\end{proposition}

In the setting of Raghavan-Thompson, fixing $p=\Theta(\log n)$ yields
an approximate path stripping routine that runs in time $\tilde{O}(m/\epsilon^{2})$.

\subsection{Submodular function minimization}

Submodular function minimization is a primitive that has been studied
in combinatorial optimization~\cite{mccormicksurvey}, and that has become increasingly more
popular in machine learning literature~\cite{suborg}. In submodular minimization
we are given a function $f$ defined over the subsets of a finite
ground set $X$, satisfying $f(S)+f(T)\geq f(S\cup T)+f(S\cap T)$
for all $S,T\subseteq X$, and have to find a subset of $X$ minimizing
the value of $f$. Examples of submodular functions include matroid
rank functions, graph cut functions, and entropy functions.

The most theoretically efficient algorithms for the problem require the use of
the ellipsoid algorithm or complex combinatorial procedures. For this reason,
practitioners typically employ the method called \emph{minimum-norm point
algorithm}, devised by
Fujishige~\cite{fujishige1984submodular,fujishige2011submodular,bach2010convex}.
Traditionally this has been implemented using variants of Wolfe's
algorithm~\cite{Wolfe76}, but no provable
convergence analysis has been available until recently. In~\cite{ChakrabartyJK14}
Chakrabarty et al. provided a robust version of Fujishige's theorem,
showing that minimizing the $2$-norm of a point in the base polyhedron
of $f$ to accuracy $1/n^{2}$ yields exact submodular function minimization.
Their approach requires $O\left(\left(n^{5}\cdot\mathcal{T}+n^{7}\right)F^{2}\right)$
time, where $\mathcal{T}$ is the time required to answer a single
function value query to $f$, and $F=\max_{i\in S\subseteq X}|f(S)-f(S\backslash\{i\})|$
is the maximum marginal difference in absolute value.

Our contribution is to show that our algorithm for the approximate
Carath\'{e}odory problem from Section~\ref{sec:linear} can replace Wolfe's
algorithm. This immediately yields an
$\tilde{O}\left(n^{6}F^{2}\cdot\mathcal{T}\right)$
time algorithm for exact submodular function minimization,
and a $\tilde{O}\left(n^{6}F^{2}/k^{4}\cdot\mathcal{T}\right)$
time algorithm for a $k$-additive approximation. The first algorithm
enjoys a better running time than Wolfe's algorithm but has a worse
oracle complexity. We leave it as an open problem if it is possible
to combine ingredients of both algorithms to achieve a running time
of $\tilde{O}\left(n^{5}F^{2}\cdot\mathcal{T}\right)$.

\begin{proposition}\label{prop:cara_submod}
Let $f:2^{X}\rightarrow\mathbb{Z}$ be a submodular function. There
exists an algorithm returning an exact minimizer of $f$ in time $\tilde{O}\left(n^{6}F^{2}\cdot\mathcal{T}\right)$,
respectively a $k$-additive approximation to the minimizer in time
$\tilde{O}\left(n^{6}F^{2}/k^{4}\cdot\mathcal{T}\right)$.
\end{proposition}

Before proving the statement, we define the base polyhedron $\mathcal{B}_{f}$
of a submodular function, and state the robust Fujishige's theorem.

\begin{definition}
Given a submodular function $f$ defined on subsets of a ground set
$X$, its base polyhedron is defined as $\mathcal{B}_{f}=\left\{ x\in\mathbb{R}^{n}:x(A)\leq f(A),\forall A\subset X,\textnormal{ and }x(X)=f(X)\right\} $.
\end{definition}

The following lemma, due to Edmonds, is required for an efficient
implementation of the step problem in mirror descent. It turns out
that in our setting, the step problem requires optimizing a linear
function over the base polyhedron $\mathcal{B}_{f}$, which can be
done efficiently using only $O(n)$ queries to $f$.

\begin{lemma}[Linear optimization over the submodular base polyhedron]\label{lem:edmonds_base_poly}
Given $c\in\mathbb{R}^{n}$, renumber the indices such that $c_{i}\leq c_{i+1}$ for
all $i$. Setting $q_{i}=f([i])-f([i-1])$ yields the solution
to the optimization problem $\min\left\{  c^\top x :x\in\mathcal{B}_{f}\right\}$.
\end{lemma}

\begin{theorem}[robust version of Fujishige's
Theorem \cite{ChakrabartyJK14}]\label{thm:chakra_submod}
Let $x\in\mathcal{B}_{f}$ such
that $ x^\top(x-z)\leq(k/n)^{2}$ for all $z\in\mathcal{B}_{f}$.
Then $x$ can be used to recover a $k$-additive approximation to
the minimizer of $f$ in $\tilde{O}(n)$ time.\end{theorem}
\begin{proof}[Proof of Proposition~\ref{prop:cara_submod}]
Consider the problem $\min_{x\in\Delta}\left\Vert Vx\right\Vert _{2}$
where the columns of $V$ are vertices of $\mathcal{B}_{f}$. Note
that the squared 2-norm of the columns of $V$ is upper bounded by
$nF^{2}$. Using the algorithm from Section~\ref{sec:linear} to solve the problem
to accuracy $(k/n)^{2}$ yields $O(n^{5}F^{2}/k^{4})$ iterations,
each of them involving a linear optimization problem over the submodular
base polyhedron. As seen in Lemma~\ref{lem:edmonds_base_poly}, this requires sorting the vector
of weights, and querying the submodular function $O(n)$ times. Hence
the algorithm runs in $\tilde{O}(n^{6}F^{2}\mathcal{T}/k^{4}$). Setting
$k$ to any constant smaller than $1$ yields an exact solution, since
by Theorem~\ref{thm:chakra_submod}, we obtain a set $S$ satisfying $f(S)\leq f(S^{*})+k<f(S^{*})+1$
and $f$ takes values over $\mathbb{Z}$.
\end{proof}

\subsection{SVM training}
Support vector machines (SVM) are an extremely popular
classification method, and have found ample usage in machine learning, with
applications ranging from finance to neuroscience. In the era of big data it is
crucial for any such method to be able to train on huge datasets. While a
number of implementations (\textsc{Liblinear}~\cite{FanCHWL08},
P-packSVM~\cite{ZhuCWZC09}, Pegasos~\cite{ShalevShwartzSSC11}) achieve excellent
convergence rates in the case of linear SVM's, handling arbitrary kernels raises
a significantly harder problem. \textsc{LibLinear} and Pegasos achieve
$O(\log(1/\epsilon))$, respectively $O(1/\epsilon)$ convergence rate, but cannot be extended beyond linear kernels. The $\epsilon$ dependence for P-packSVM
scales as $O(1/\epsilon)$, but it requires knowing the Cholesky factorization of
the kernel matrix in advance. In our case, a simple extension of the method
described in Section~\ref{sec:linear} gives $O(1/\epsilon^2)$ convergence, while
only requiring matrix-vector multiplications involving the kernel matrix. So
Cholesky factorization is no longer required, and the matrix does not need to be
stored explicitly. In the case of linear SVM's, our method runs in nearly linear
time.

Our approach is inspired from a reformulation of the training problem of
Kitamura, Takeda, and Iwata~\cite{KitamuraTI14}, who present a method for SVM
training based on Wolfe's algorithm. Their algorithm relies on a dual
formulation introduced by Sch\"{o}lkopf et al.~\cite{ScholkopfSWB00} which can
be easily reformulated as a convex problem over a product of two convex sets.
More specifically, we are given empirical data
$(x_{i,}y_{i})\in\mathcal{X}\times\left\{ \pm1\right\} $, $1\leq i\leq n$, along
with a function that maps features to a Hilbert space
$\Phi:\mathcal{X}\rightarrow\mathbb{\mathcal{H}}$, which determines a kernel
function $k(x,y)=\langle\Phi(x),\Phi(y)\rangle$.  Let $K\in\mathbb{R}_{n\times
n}$, where $K_{ij}=k(x_{i,}x_{j})$, $E_{+}=\left\{ e_{i}:y_{i}=+1\right\} $,
$E_{-}=\left\{ e_{i}:y_{i}=-1\right\}$.

In~\cite{KitamuraTI14}, the $\nu$-SVM problem is reformulated as:

\begin{align*}
\min                              &\quad \left(\lambda_{+}-\lambda_{-}\right)^\top K\left(\lambda_{+}-\lambda_{-}\right)\\
\textnormal{subject to} &\quad \lambda_{+}\in \RCH_{\eta}(E_{+})\\
                                    &\quad \lambda_{-}\in \RCH_{\eta}(E_{-})
\end{align*}

where $\eta=\frac{2}{\nu n}$ and $\RCH_{\eta}(A):=\left\{ \sum_{a\in
A}\lambda_{a}a | 0\leq\lambda_{a}\leq\eta,\sum_{a\in A}\lambda_{a}=1\right\} $
is the \emph{restricted convex hull} of set $A$.

Our approach to solve this problem will be to rephrase it as a saddle point
problem (similar to what was done for the approximate Carath\'{e}odory problem)
and apply Mirror Descent, with a suitable Mirror Map, to solve the dual.
Before doing that, we introduce a few useful definitions and facts:

\begin{definition}
Let $K$ be a symmetric positive definite matrix. Then $\left\Vert x\right\Vert _{K}:=\sqrt{x^{\top}Kx}$.
\end{definition}

\begin{proposition}
The dual norm of $\left\Vert x\right\Vert _{K}$ is $\left\Vert x\right\Vert
_{K^{-1}}$.
In other words $\left\Vert x\right\Vert _{K}=\max_{y:\left\Vert y\right\Vert
_{K^{-1}}\leq1}\left\langle y,x\right\rangle $.\end{proposition}

\begin{proof}
  This can be verified using Lagrange multipliers: over the unit
  $\norm{\cdot}_{K^{-1}}$-ball the term $y^\top x$ attains its maximum at 
  $y=Kx/\left\Vert Kx\right\Vert _{K^{-1}}=Kx/\sqrt{x^{\top}Kx}$.
We can verify that for this choice of $y$, $y^\top x =\frac{x^{\top}Kx}{\sqrt{x^{\top}Kx}}=\Vert x\Vert_{K}$.
\end{proof}

\begin{definition} 
Let $\mathcal{S_{\eta}}=\left\{ \lambda_{+}-\lambda_{-} | \lambda_{+}\in \RCH_{\eta}(E_{+}),\lambda_{-}\in \RCH_{\eta}(E_{-})\right\} $.
\end{definition}

\begin{proposition}[Linear optimization over $S_{\eta}$]\label{prop:lin_opt_s_eta}
 Linear optimization over $S_{\eta}$ can be implemented in $\tilde{O}(n)$ time.
\end{proposition}
\begin{proof}
The implementation of the linear optimization routine is done in near-linear
time via a simple greedy algorithm. The first thing to notice is that the objective is separable, so it is sufficient to optimize
separately on $E_+$ and $E_-$. 
This can be done easily, since we need to distribute one unit of
mass over the coordinates that span $E_{+}$ and one unit of mass over
the coordinates that span $E_{-}$, such that no coordinate receives
more than $\eta$. Therefore adding mass to the coordinates spanning
$E_{+}$ in increasing order of the weights $y$, and vice-versa to
those spanning $E_{-}$ yields the optimal solution.
\end{proof}

With these facts on hand, we can now proceed to describing our equivalent
formulation as a saddle-point problem, which will then be solved using
a similar method to the one we employed for the previous applications. 

Note that instead of directly using the kernel matrix $K$ in the
formulation, we replace it with $\tilde{K}=K+\mbox{\ensuremath{\frac{\epsilon}{2}}}I$.
This only changes the value of the objective by at most $\epsilon/2$ and it has
the advantage of making $\tilde{K}$ positive semidefinite, since it is now
guaranteed to be non-degenerate. This allows us to write the objective function
$(\lambda_+ - \lambda_-)^\top \tilde{K} (\lambda_+ - \lambda_-)$ as
$\norm{\lambda_+ - \lambda_-}_{\tilde{K}}^2$. This formulation can be easily
converted to a saddle point problem:

$$ \min_{\lambda\in\mathcal{S}}\Vert\lambda\Vert_{\tilde{K}} 
=\min_{\lambda\in\mathcal{S_{\eta}}}\max_{y:  \Vert
y\Vert_{\tilde{K}^{-1}}\leq1} y^\top \lambda = 
-\min_{y:\Vert
y\Vert_{\tilde{K}^{-1}}\leq1}\left(-\min_{\lambda\in\mathcal{S_{\eta}}}
y^\top \lambda\right) = -\min_{y:\Vert y\Vert_{\tilde{K}^{-1}}\leq1}f(y) $$
for $f(y) := -\min_{\lambda\in\mathcal{S_{\eta}}}  y^\top \lambda $ defined
over the $\norm{\cdot}_{\tilde{K}^{-1}}$-ball.

The subgradients of $f$ are easy to compute, since they require a
simple linear optimization over $\mathcal{S}$: $$\partial f(y)= -
\arg\min_{\lambda\in\mathcal{S}} y^\top \lambda$$
which can be done in time $\tilde{O}(n)$ using the greedy algorithm
described in Proposition~\ref{prop:lin_opt_s_eta}. The mirror map of choice for the domain
$\left\{ y:\Vert y\Vert_{\tilde{K}^{-1}}\leq1\right\} $will be $\omega:\left\{ y:\Vert y\Vert_{\tilde{K}^{-1}}\leq1\right\} \rightarrow\mathbb{R}$, $\omega(y)=\frac{1}{2}\Vert y\Vert_{\tilde{K}^{-1}}^{2}$,
with
$$\omega^{*}(z)=\begin{cases}
\frac{1}{2}\Vert z\Vert_{\tilde{K}}^{2} & \textnormal{if }\Vert z\Vert_{\tilde{K}}\leq1\\
\Vert z\Vert_{\tilde{K}}-\frac{1}{2} & \textnormal{if }\Vert z\Vert_{\tilde{K}}>1
\end{cases}.$$ Also, similarly to
Proposition~\ref{prop:fenchel_dual_computation}, we
have $\nabla\omega^{*}(z)=\tilde{K}z\cdot\min(1,1/\Vert z\Vert_{\tilde{K}})$,
hence $\Vert\nabla\omega^{*}(Z)\Vert_{\tilde{K}^{-1}}\leq1$.

The only thing left to do is to analyze the algorithm's iteration
count by bounding the strong convexity of $\omega$ and the Lipschitz
constant of $f$. We will do this with respect to $\Vert\cdot\Vert_{2}$.

\begin{proposition}
$\omega$ is $\min\left(\frac{\epsilon}{2},\left(\Vert K\Vert+\frac{\epsilon}{2}\right)^{-1}\right)$-strongly
convex with respect to $\Vert\cdot\Vert_{2}$, where $\Vert K\Vert$ is
the spectral norm of $K$.
\end{proposition}

\begin{proof}
Writing down the Hessian of the mirror map, we obtain $\nabla^{2}\omega(y)=\tilde{K}^{-1}=\left(K+\frac{\epsilon}{2}I\right)^{-1}\succeq\min\left(\frac{\epsilon}{2},\left(\Vert K\Vert+\frac{\epsilon}{2}\right)^{-1}\right)I$.
The reason for using $\tilde{K}$ instead of $K$ in the formulation
is now evident: if $K$ is not full rank, then $\omega$ is not strongly
convex. Adding a small multiple of the identity forces all the eigenvalues
of $\tilde{K}$ to be at least $\epsilon/2$, and avoids the degeneracy
where some of them may be zero.
\end{proof}

\begin{proposition}
$f$ is $2\sqrt{\eta}$-Lipschitz with respect to $\Vert\cdot\Vert_{2}$.
\end{proposition}
\begin{proof}
We simply need to bound the $2$-norm of the subgradient. By the construction
presented in Proposition~\ref{prop:lin_opt_s_eta} the subgradient contains $2\cdot\lceil1/\eta\rceil$
nonzero coordinates, $2\cdot\lfloor1/\eta\rfloor$ of which are precisely
$\eta$. This enables us to obtain a better upper bound than one would
usually expect on the $2$-norm of the subgradient, namely $\sqrt{2\cdot\left(\lfloor1/\eta\rfloor\cdot\eta^{2}+(1-\eta\cdot\lfloor1/\eta\rfloor)^{2}\right)}\leq\sqrt{2\cdot(2/\eta)\cdot\eta^{2}}=2\sqrt{\eta}$.
\end{proof}

\begin{proposition}
$\max_{y:\Vert y\Vert_{\tilde{K}}\leq1}\frac{1}{2}\Vert y\Vert_{\tilde{K}}^{2}\leq\frac{1}{2}$
\end{proposition}

Finally we can put everything together:
\begin{theorem}
An $\epsilon$-approximate solution to $\nu$-SVM can be found in
$O\left(\eta\cdot\max\left(\frac{2}{\epsilon}\Vert K\Vert+\frac{\epsilon}{2}\right)/\epsilon^{2}\right)=O\left(\max\left(\frac{1}{\epsilon}\Vert K\Vert\right)/\left(\nu n\epsilon^{2}\right)\right)$
iterations.
\end{theorem}
\begin{proof}
Follows from plugging in the parameters $\sigma=\min\left(\epsilon/2,\left(\Vert K\Vert+\epsilon/2\right)^{-1}\right)$,
$L=2\sqrt{\eta}$, $R=O(1)$ into the mirror descent algorithm.
\end{proof}

At this point, it makes sense to analyze the performance of our algorithm
for the most common choices of SVM kernels, which only requires bounding
the spectral norm of the kernel matrix; for this purpose we will simply
use the trace bound. The results are summarized in the table below.
The last column of the table contains the number of iterations required
to find a solution down to a precision of $\epsilon$, given that
all the vectors $x_{i}$ belong to the unit $\ell_{2}$ ball.

\begin{center}
\hspace*{-1cm}
\begin{tabular}{|c|c|c|}
\hline 
Kernel type & Upper bound on $\Vert K\Vert$ & Iteration count\tabularnewline
\hline\hline
\shortstack{Polynomial (homogeneous): \\ $K_{ij}=\langle x_{i},x_{j}\rangle^{d}$} & $n\cdot\max_{i}\Vert x_{i}\Vert_{2}^{2d}$ & $O\left(\max\left(\frac{1}{n\nu\epsilon^{3}},\frac{1}{\nu\epsilon^{2}}\right)\right)$\tabularnewline
\hline 
\shortstack{Polynomial (inhomogeneous): \\ $K_{ij}=\left(1+\langle
x_{i},x_{j}\rangle\right)^{d}$} & $n\cdot\left(1+\max_{i}\Vert x_{i}\Vert_{2}^{2}\right){}^{d}$ & $O\left(\max\left(\frac{1}{n\nu\epsilon^{3}},\frac{2^{d}}{\nu\epsilon^{2}}\right)\right)$\tabularnewline
\hline 
\shortstack{RBF: \\ $K_{ij}=\exp(-\Vert x_{i}-x_{j}\Vert^{2}/2\sigma^{2})$} & $n$ & $O\left(\max\left(\frac{1}{n\nu\epsilon^{3}},\frac{1}{\nu\epsilon^{2}}\right)\right)$\tabularnewline
\hline 
\shortstack{Sigmoid: \\ $K_{ij}=\tanh(\alpha \cdot \langle x_i, x_j \rangle + c)$} & $n$ & $O\left(\max\left(\frac{1}{n\nu\epsilon^{3}},\frac{1}{\nu\epsilon^{2}}\right)\right)$\tabularnewline
\hline 
\end{tabular}
\end{center}

It is worth mentioning that each iteration requires $\tilde{O}(n)$
time for computing the subgradient, and a multiplication of the kernel
matrix with a vector; one advantage is that the kernel matrix does
not need to be explicitly stored, as its entries can be computed on
the fly, whenever needed. In the case of linear kernels, this computation
is implemented in linear time since $\tilde{K}z=\left[x_{1}\vert\dots\vert x_{n}\right]^{\top}\left[x_{1}\vert\dots\vert x_{n}\right]z+\frac{\epsilon}{2}z$,
which requires computing a linear combination $h=\sum_{i}x_{i}\cdot z_{i}$
of the vectors $x$, and $n$ dot products between vectors from the
training set and $h$.